\documentclass[10pt]{article}

\usepackage[margin=1in,includefoot,footskip=30pt]{geometry} 
\usepackage{standalone} 							
\usepackage[backend=bibtex,style=nature]{biblatex} 	
\usepackage[colorlinks,bookmarks=false]{hyperref} 	
\usepackage{enumerate} 								
\usepackage{enumitem} 						
\usepackage{amsfonts} 						
\usepackage{braket} 						
\usepackage{amssymb} 						
\usepackage{amsmath} 						
\usepackage{amsthm} 						
\usepackage{mathtools} 						
\usepackage{authblk} 						
\usepackage{framed} 						
\usepackage{varwidth,float} 				
\usepackage{etoolbox} 						
\usepackage[table,xcdraw,svgnames]{xcolor} 	
\usepackage{float} 							
\usepackage[titletoc,toc,title]{appendix}
\usepackage{afterpage}
\usepackage[ruled,vlined,noend]{algorithm2e}	

\newtheorem{theorem}{Theorem}[section]
\newtheorem{lemma}[theorem]{Lemma}
\newtheorem{remark}[theorem]{Remark}

\newfloat{box_fig}{htbp}{box_fig}\floatname{box_fig}{Box}

\AtBeginDocument{%
  \hypersetup{
    citecolor=NavyBlue,
    linkcolor=Red,   
    urlcolor=Purple}}
    
\bibliography{SPF_library}

\newbibmacro{string+doi}[1]{%
  	\iffieldundef{doi}{
  		\iffieldundef{url}{#1}{\href{\thefield{url}}{#1}}
	}{\href{http://dx.doi.org/\thefield{doi}}{#1}}}
\DeclareFieldFormat{title}{\usebibmacro{string+doi}{\mkbibemph{#1}}}
\DeclareFieldFormat[article]{title}{\usebibmacro{string+doi}{\mkbibquote{#1}}}

\makeatletter
\pretocmd\start@align{%
  \if@minipage\kern-\topskip\kern-\abovedisplayskip\fi
}{}{}
\makeatother

\newcommand{\qsupp}[1]{\mathcal{Q}(#1)}

\newcommand{\acomm}[2]{\{#1, #2\}}
\newcommand{\comm}[2]{\big[#1, #2\big]}
\newcommand{\abs}[1]{\mathopen|#1\mathclose|}

\newcommand{\gen}[1]{\left\langle#1\right\rangle}
\newcommand{\id}{\mathbb{I}}
\newcommand{\eval}[1]{{\left\langle#1\right\rangle}}
\newcommand{\CZ}[0]{\mathrm{CZ}}
\newcommand{\red}[1]{{\color{red} #1 }}

\DeclarePairedDelimiter\floor{\lfloor}{\rfloor}

\begin{document}

\title{Loss-tolerant teleportation on large stabilizer states}

\date{}

\author[1]{Sam Morley-Short\thanks{\href{sam.morley-short@bristol.ac.uk}{sam.morley-short@bristol.ac.uk}}}
\author[1,2,3]{Mercedes Gimeno-Segovia}
\author[2]{Terry Rudolph}
\author[1]{Hugo Cable}

\renewcommand\Authfont{\fontsize{12}{14.4}\selectfont}
\renewcommand\Affilfont{\fontsize{9}{10.8}\itshape}

\affil[1]{Quantum Engineering Technology Labs, H.\ H.\ Wills Physics Laboratory and Department of Electrical and Electronic Engineering, University of Bristol, BS8 1FD, UK}
\affil[2]{Department of Physics, Imperial College London, London SW7 2AZ, UK}
\affil[3]{Institute for Quantum Science and Technology, University of Calgary, Alberta T2N 1N4, Canada}

\maketitle
\begin{abstract}
We present a general method for finding loss-tolerant teleportation on large, entangled stabilizer states using only single-qubit measurements, known as \emph{stabilizer pathfinding} (SPF).
For heralded loss, SPF is shown to generate optimally loss-tolerant measurement patterns on any given stabilizer state.
Furthermore, SPF also provides highly loss-tolerant teleportation strategies when qubit loss is unheralded.
We provide a fast algorithm for SPF that updates continuously as a state is generated and measured, which is therefore suitable for real-time implementation on a quantum-computing device.
When compared to simulations of previous heuristics for loss-tolerant teleportation on graph states, SPF provides considerable gains in tolerance to both heralded and unheralded loss, achieving a near-perfect teleportation rate ($> 95\%$) in the regime of low qubit loss ($< 10\%$) on various graph state lattices.
Using these results we also present evidence that points towards the existence of loss-tolerant thresholds on such states, which in turn indicates that the loss-tolerant behaviour we have found also applies as the number of qubits tends to infinity.
Our results represent a significant advance towards the realistic implementation of teleportation in both large-scale and near-future quantum architectures that are susceptible to qubit loss, such as linear optical quantum computation and quantum communication networks.

\end{abstract}

\section{Introduction} \label{sec:introduction}


Many new quantum technologies demand the teleportation of quantum states across large, multiparty entangled states \cite{Pirandola2015,Kimble2008,Epping2016,Simon2017,Das2018,Gimeno-Segovia2015}.
A common example is provided by measurement-based quantum computation (MBQC) \cite{Raussendorf2001,Raussendorf2003}, which uses single-qubit measurements on cluster states and feed-forward of measurement outcomes to implement universal quantum computation.
Teleportation steps are used extensively in MBQC, whether following the original proposal \cite{Raussendorf2001} or generalisations using alternative entangled resource states \cite{Gross2010}.
In practise, any protocol for quantum computation (or related applications such as in quantum communications \cite{Azuma2015}) must also tolerate qubit dephasing and loss.
While the primary source of error for many quantum computing platforms is qubit dephasing, loss errors are known to dominate in architectures such as linear optical quantum computation (LOQC) \cite{Knill2001,Kok2007,Gimeno-Segovia2015,Li2015}.
Currently, the main approach to mitigating significant degrees of loss are quantum error correcting codes (QECC) \cite{Stace2009}, loss-tolerant qubit encodings \cite{Varnava2006,Varnava2007,Varnava2008}, or some other process imposing additional resource costs, such as the proposal of \cite{Campbell2008} which enables photon loss to be converted into a linear time cost, providing successful quantum gates within a modular light-matter based architecture.

In this work we present a new method for teleportation that exploits the correlations of large, entangled stabilizer states using only single-qubit measurements, known as \emph{stabilizer pathfinding} (SPF).
For heralded loss, we show that SPF provides optimally loss-tolerant measurement patterns for all stabilizer states, as well as tolerance of unheralded qubit loss.
To implement SPF in a realistic setting, we also provide an algorithm that can generate SPF measurement patterns with low computational overhead based on applying minimal updates during states generation and measurement.

When compared to simulations of previous heuristics for teleportation on quantum graph states, SPF provides significant gains in loss tolerance for both the heralded and unheralded case.
For example, when applied to the square-lattice graph states (i.e.\ cluster states) commonly used for MBQC, we find that SPF achieves a teleportation rate of $T \approx 98\%$ for $10\%$ heralded qubit loss, compared to $T \approx 40\%$ using previous teleportation techniques based on localisable entanglement \cite{Raussendorf2001,Hein2006}.
When the loss is unheralded on the same state, SPF measurement strategies also achieve at least $T \approx 84\%$---where there was no previously-known method for achieving loss tolerance for teleportation.

We also provide evidence of critical loss-tolerant thresholds on a variety of graph state lattices.
These would show that loss-tolerant teleportation can be achieved in the limit of infinite lattice size, with existence of loss-tolerant measurement patterns guaranteed below some threshold loss rate.
Our results provide an optimistic outlook on the reduction of loss rates in quantum computation and communication architectures as well as ensuring optimal use of intermediately-sized states generated by near-term devices.


The paper is structured as follows.
Section \ref{sec:motivation} motivates our work by considering the task of teleportation on stabilizer states and presents previous approaches to achieving loss tolerance.
The stabilizer pathfinding approach to teleportation is then presented in section \ref{sec:SPF} which outlines an algorithm for it's computation.
Our main results are given in section \ref{sec:loss_tol} which provides numerical simulations to highlight SPF's improved loss tolerance in the case of both heralded and unheralded loss.
Section \ref{sec:discussion} then discusses SPF's algorithmic efficiency and it's implications for LOQC and other quantum technology platforms.
Finally, section \ref{sec:conclusion} summarises the work and suggests a selection of avenues for further research.

\section{Background and motivation} \label{sec:motivation}


We now present a short introduction to teleportation on stabilizer states followed by an example to motivate the need for a general approach for finding teleportation measurement patterns.
In what follows we will assume familiarity with the standard definitions on the stabilizer formalism, graph states and MBQC and refer the reader to \cite{Gottesman1997,Hein2004,Raussendorf2003} for more details.
Also given the equivalence between stabilizer and graph states \cite{VandenNest2004,Anders2006}, we shall only consider graph states here but note that the following applies to stabilizer states.
An introduction to stabilizer formalism is also provided in Appendix \ref{apx:stab_theory}.

\subsection{Teleportation on stabilizer states} \label{sec:motivation:teleportation}


Consider an arbitrary quantum state $\ket{\psi}$ on \emph{input} qubit $I$ with logical operators $\bar{X}_\psi = X_I, \bar{Z}_\psi = Z_I$.
Now consider the entangling of $\ket{\psi}$ with $n$ other qubits in some graph state such that the resultant state $\ket{\Psi}$ is now defined by a pair of logical operators $\bar{X}_\Psi, \bar{Z}_\Psi$ and stabilizer generators $\mathcal{G}_\Psi = \{K_i\}^n_{i=1}$ that form the closed group $\mathcal{S}_\Psi = \gen{\mathcal{G}_\Psi}$ of all stabilizers of $\ket{\Psi}$ under multiplication.
Teleportation on $\ket{\Psi}$ aims to find some set of single-qubit measurements or \emph{measurement pattern} $M$ that recovers $\ket{\psi}$ on some \emph{output} qubit $O$, or equivalently, that produce two anti-commuting logical operators acting only on $O$.
Qubits not measured by any element of $M$ can then be lost without impeding teleportation, such that maximal loss tolerance is achieved by minimising $\abs{M}$.
Hence, the set of all teleportation protocols which can tolerate some amount of loss can be known by finding all $M$ that omit at least one qubit.


We now present a general method for finding valid $M$ on $\ket{\Psi}$.
First, recall that any product of the logical operator and stabilizer is also a logical operator on $\ket{\Psi}$, thereby defining the set of all logical operators $\mathcal{L}_\Psi = \gen{\bar{X}_\Psi, \bar{Z}_\Psi} \times \mathcal{S}_\Psi$.
Given a pair of logical operators $\bar{X}, \bar{Z} \in \mathcal{L}_\Psi$ such that
\begin{align}
	\acomm{\bar{X}^{[O]}}{\bar{Z}^{[O]}} = 0, \quad\textrm{and}\quad \comm{\bar{X}^{[a]}}{\bar{Z}^{[a]}} = 0 \;\;\forall\; a \neq O, \label{eq:motivation:log_op_pairs}
\end{align}
where $A^{[i]}$ denotes the Pauli operator of $A$ acting on qubit $i$, then it is easy to see that the single-qubit measurement of all $\bar{X}^{[a]}, \bar{Z}^{[a]} \neq \id$ will achieve teleportation onto $O$.
Specifically, the measurement pattern produced by the pair of logical operators $\bar{X}$ and $\bar{Z}$ is given by
\begin{align}
	M_{\bar{X}, \bar{Z}} = \{\bar{X}^{[i]}: \,\bar{X}^{[i]} \neq \id, \;\forall\; i \neq O\} \cup \{\bar{Z}^{[i]}: \,\bar{Z}^{[i]} \neq \id, \;\forall\; i \neq O\},
\end{align}
which has \emph{weight} $w = \abs{M_{\bar{X}, \bar{Z}}}$.
The set of all valid measurement patterns $\mathcal{M} = \{M_{\bar{X}, \bar{Z}}\}$ is then given by finding all logical operator pairings satisfying equation \eqref{eq:motivation:log_op_pairs}.
Given the equivalence between states' logical operators and stabilizers, we refer to this method for teleportation as \emph{stabilizer pathfinding} (SPF).
From the above requirements we define the \emph{stabilizer pathfinding conditions}, which are summarised in box \ref{box:spf_conds}.

\begin{box_fig}[t]
	\begin{framed}
		{\bf STABILIZER PATHFINDING CONDITIONS:} \vspace{4pt}
		
		Consider the state $\ket{\Psi}$ defined by logical operators $\mathcal{L}_\Psi$ that encodes a single logical qubit state $\ket{\psi}$.
		A valid measurement pattern that recovers $\ket{\psi}$ on qubit $O$ of $\ket{\Psi}$ can be found from any pair of logical operators $\bar{X}, \bar{Z} \in \mathcal{L}_\Psi$ that:
		\vspace{-2pt}
		\begin{enumerate}[label=\alph*)]
			\itemsep0em 
			\item \emph{anticommute} on qubit $O$, and
			\item \emph{commute} on each qubit which is not $O$.
		\end{enumerate}
		Given these conditions are satisfied, teleportation is achieved by performing the set of single-qubit measurements represented by each non-identity Pauli operator of $\bar{X}, \bar{Z}$ on all qubits other than $O$.
	\end{framed}
	\vspace{-7pt}
	\caption{Conditions any pair of logical operators must satisfy to provide a teleportation measurement pattern.}
	\label{box:spf_conds}
\end{box_fig}


Given the significant number of $\bar{X}, \bar{Z}$ pairs for large states, measurement patterns are often found from heuristic methods.
The most common heuristic for finding a subset of $\mathcal{M}$ on graph states is a technique we shall refer to as \emph{graph pathfinding} (GPF), originally proposed for teleportation in MBQC and producing localizable entanglement \cite{Raussendorf2001,Hein2006}.
As used by MBQC on graph states, this approach requires finding a path $P = \{I, \ldots, O\}$ between qubits $I$ and $O$ and $P$'s graph neighbourhood $\Pi$ (all qubits that neighbour a qubit in $P$ that are not themselves in $P$), on which single-qubit $X$ and $Z$ measurements are performed respectively.
Finding $M$ for loss-tolerant teleportation is thus achieved by minimising $\abs{P \cup \Pi}$.
The graph pathfinding heuristic is usually understood by observing that teleportation occurs from $X$ measurements along a linear graph state between $I$ and $O$ produced from the $Z$ measurements.

Equally, by recalling that graph state's generators are given by $K_i = X_i \bigotimes_{j \in N_G(i)} Z_j \;\;\forall\;\; i=1,\ldots,n$ (where $N_G(i)$ is the neighbourhood of $i$ on graph $G$), it is easy to see why such a technique works through the lens of stabilizer pathfinding.
Specifically, given $P$ there are two always logical operators $\bar{X}$, $\bar{Z}$ with $X$ operators at odd and even positions along $P$ respectively, either terminating with $Z_O$ for $\bar{X}$ when $\abs{P}$ is odd or for $\bar{Z}$ when $\abs{P}$ is even, with $Z$ operators on qubits in $\Pi$.
When paired such logical operators then give the usual $M$ for graph pathfinding.

\subsection{Limitations of graph pathfinding} \label{sec:motivation:limitations_of_gpf}


We now present a motivating example for the relevance of stabilizer pathfinding to loss-tolerant teleportation.
Consider the state $\ket{\Psi}$, depicted below:
\begin{figure}[H]
	\begin{center}
		\begin{varwidth}{0.35\linewidth}
			\begin{center}	
				\begin{figure}[H]
				    \includegraphics[width=\textwidth]{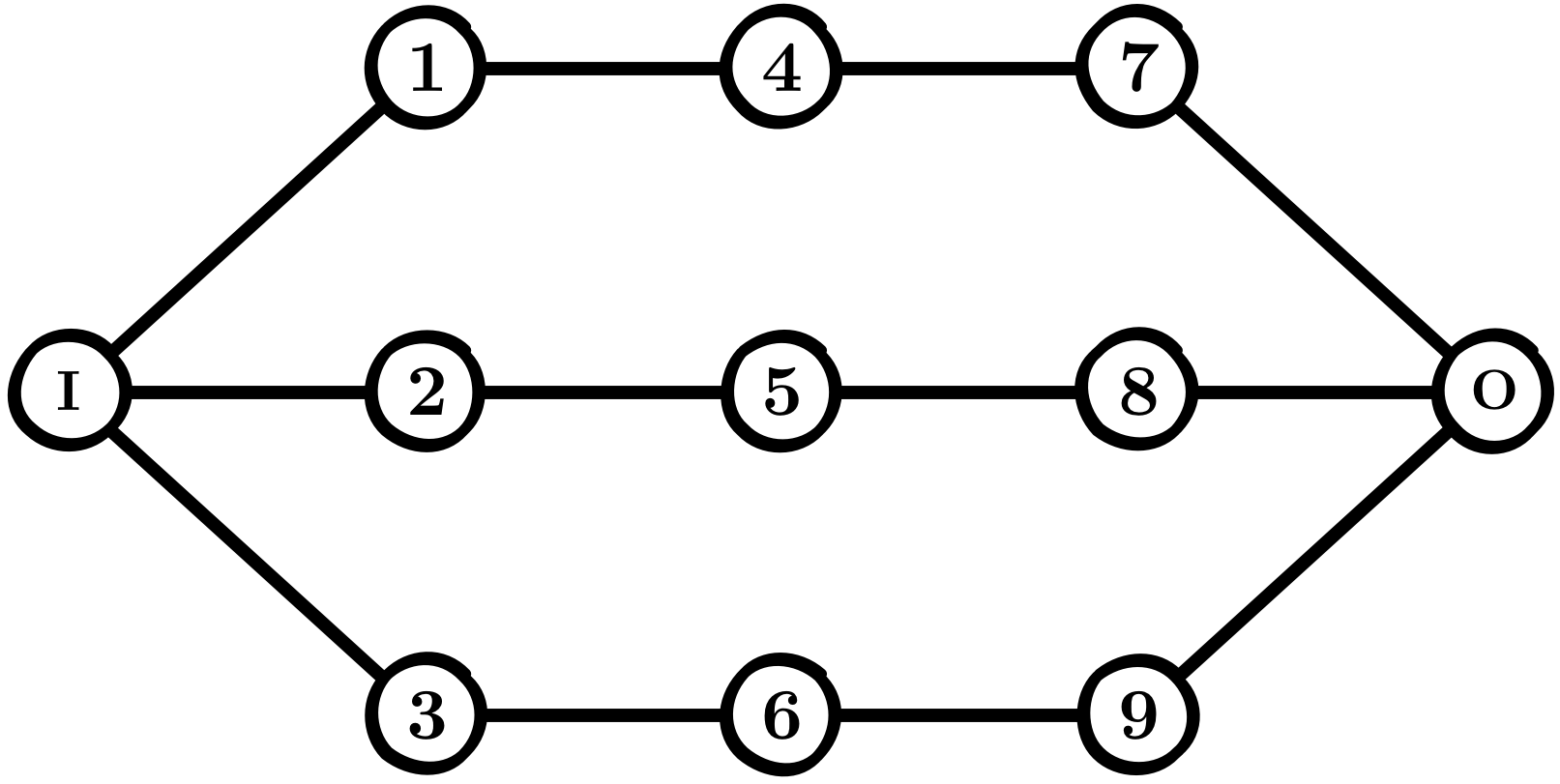}
				\end{figure}
			\end{center}
		\end{varwidth}%
		\hspace{20pt}
		\begin{minipage}{0.45\linewidth}
			\vspace{-10pt}
			\begin{align}	
				\arraycolsep=1.4pt
				\begin{array}{rlll}
					\bar{X}_\Psi = &\;\, Z_I \\
					\bar{Z}_\Psi = &\;\, X_IZ_1Z_2Z_3 \\
			\mathcal{G}_{\Psi} = &\{K_1 = Z_IX_1Z_4, & K_2 = Z_IX_2Z_5, & K_3 = Z_IX_3Z_6, \\
						 	         &\;\, K_4 = Z_1X_4Z_7, & K_5 = Z_2X_5Z_8, & K_6 = Z_3X_6Z_9, \\
		  					         &\;\, K_7 = Z_4X_7Z_O, & K_8 = Z_5X_8Z_O, & K_9 = Z_6X_9Z_O, \\
						   	         & \multicolumn{2}{l}{\;\,K_O = Z_7Z_8Z_9X_O \}}
				\end{array} \notag
			\end{align}
		\end{minipage}
		\vspace{-22pt}
	\end{center}
	\vspace{-5pt}
	\caption{An example graph state on which teleportation is to be performed from input qubit $I$ to output qubit $O$, defined by logical operators $\bar{X}_\Psi, \bar{Z}_\Psi$ and stabilizer generators $\mathcal{G}_\Psi$.}
	\label{fig:three_ways_1}
\end{figure}

On the above example, graph pathfinding clearly provides only three measurement patterns, such as $P = \{I, 2, 5, 8, O\}$, and thus provides tolerance to the loss of at most (but not any) two qubits, such as $\{4, 6\}$, with the associated $M$ depicted in figure \ref{fig:both_multi_paths}a).
Furthermore, since each $M$ associated with a path contains anticommuting measurements on at least one qubit, there is little-to-no ability to switch between them in the case of unheralded loss.


Now consider an alternative set of three measurement patterns provided by stabilizer pathfinding:
\begin{align*}
	\bar{X} = K_1 K_7\bar{X}_\Psi = X_1 X_7 Z_O, \; \bar{Z} = K_4 K_5 K_6 K_O\bar{Z}_\Psi = X_I X_4 X_5 X_6 X_O \;\Rightarrow\; M_1 &= \{X_I, X_4, X_5, X_6, X_1, X_7\}\\
	\bar{X} = K_2 K_8\bar{X}_\Psi = X_2 X_8 Z_O, \; \bar{Z} = K_4 K_5 K_6 K_O\bar{Z}_\Psi = X_I X_4 X_5 X_6 X_O \;\Rightarrow\; M_2 &= \{X_I, X_4, X_5, X_6, X_2, X_8\}\\
	\bar{X} = K_3 K_9\bar{X}_\Psi = X_3 X_9 Z_O, \; \bar{Z} = K_4 K_5 K_6 K_O\bar{Z}_\Psi = X_I X_4 X_5 X_6 X_O \;\Rightarrow\; M_3 &= \{X_I, X_4, X_5, X_6, X_3, X_9\}
\end{align*}
as depicted in figure \ref{fig:both_multi_paths}b).
There are two key differences between these $M$ and those provided by graph pathfinding.
Firstly, each $M$ can tolerate twice the amount of lost qubits, equating to a four-fold increase in the number qubit loss configurations tolerable.
Secondly, since no two patterns require contradictory measurements on any qubit, the attempt of one pattern does not preclude the later attempt of another.
Although the latter difference is irrelevant in the case of heralded qubit loss, this fact crucially allows tolerance of unheralded loss events.
For example, consider we begin a teleportation protocol by the successful measurement of $X_I, X_4, X_5$, and $X_6$, leaving three possible sets of measurements: $\{X_1, X_7\}$, $\{X_2, X_8\}$, and $\{X_3, X_9\}$.
Since only one pair must succeed, any loss on up to two pairs can be tolerated as long as one is completed\footnote
{
	In the case that there is no additional cost to extraneous measurements, each pair can be measured simultaneously.
}.
This can also be seen by noting that if any pair is successfully measured, any remaining (and potentially lost) qubits are disentangled from the final state on qubit $O$.


From the above it is clear the measurement patterns provided by graph pathfinding represent only a small fraction of all $M \in \mathcal{M}$.
For example, when stabilizer pathfinding is applied on the previous state we find $\abs{\mathcal{M}} = 2657$, allowing 60 different combinations of lost qubits, with at most four qubits left unmeasured.
However, finding the set $\mathcal{M}$ through an exhaustive search is impractical for large states in general.
Furthermore, many, if not the majority of $M \in \mathcal{M}$ will not tolerate any qubit loss.
In order to overcome this challenge, we shall now present algorithm that finds all maximally lost-tolerant\footnote
{
	Here \emph{maximally lost-tolerant} refers to the fact that our algorithm will return measurement patterns in descending loss tolerance, finding those measurement patterns that are tolerant to the greatest number of qubits first.
}
 measurement patterns without any exhaustive searches.

\section{Stabilizer pathfinding} \label{sec:SPF}


Given that $O(2^{2n})$ possible pairs of logical operators exist for a state with $n$ generators, computing $\mathcal{M}$ by brute force is clearly impractical for even modestly sized states\footnote
{
	There are similarly $O(4^n)$ possible Pauli measurement patterns on $n$ qubits, providing an equally impractical computation.
}.
The most practical aspect of our work is an algorithm that implements stabilizer pathfinding to find loss-tolerant measurement patterns without the need for exhaustive searches.


Functionally our algorithm is divided into two distinct subroutines: i) finding all stabilizers of the state that are relevant for teleportation, and ii) finding all pairs of logical operators that produce maximally loss-tolerant measurement patterns.
In this section we provide an outline of each routine's challenges and our solutions, with full technical details found in Appendix\ \ref{apx:SPF}, including a full pseudocode description in algorithm \ref{alg:SPF}.
Readers primarily concerned with the degree of loss tolerance afforded by stabilizer pathfinding are directed to section \ref{sec:loss_tol}.

\subsection{Which stabilizers are relevant for teleportation?} \label{sec:SPF:relevant_stabs}

 
To prevent the need to store and update all $2^n$ stabilizers, we now consider which of a state's stabilizers are relevant to teleportation.
This will allow the identification of the subset of stabilizers that must be tracked for stabilizer pathfinding.

\subsubsection{Logical operators as combinations of stabilizer generators} \label{sec:SPF:relevant_stabs:stab_combos}


Consider an arbitrary state $\ket{\Psi}$ with stabilizers $\mathcal{S}_\Psi = \gen{\mathcal{G}_\Psi}$, where $\mathcal{G}_\Psi = \left\{K_i\right\}_{i=1}^n$.
Given that $\mathcal{S}_{\Psi}$ form a closed group under multiplication, we can label each stabilizer $S_c \in \mathcal{S}_{\Psi}$ by the set of generator indices $c$ from which it is produced, such that
\begin{equation}
	S_c = \prod_{i \in c}K_i\;.
\end{equation}
We shall refer to $c$ as the stabilizer's generator \emph{combination}, by which it is uniquely
 defined (given a fixed $\mathcal{G}_\Psi$).

\begin{figure}[t]
	\centering
	\vspace{0pt}
	\includegraphics[width=0.7\textwidth]{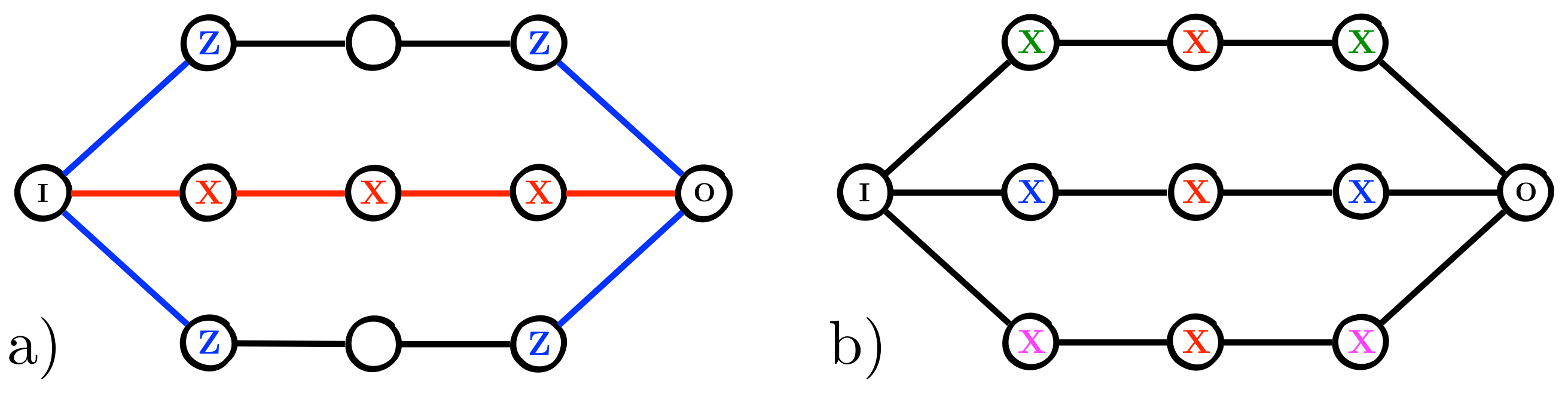}
	\caption{
		Possible measurement patterns for teleportation between qubits $I$ and $O$ provided by a) the graph pathfinding heuristic and b) our generalised stabilizer pathfinding (where measurement $X_I$ is also needed in both cases).
		In a), all measurements must be successfully completed for teleportation, providing a loss tolerance to the two unlabelled qubits (with the associated path and neighbourhood highlighted in red and blue respectively).
		In b), if the centre column of qubits are successfully measured then teleportation is completed by the successful measurement of both qubits in any of the three pairs of the same colour.
		We note that stabilizer pathfinding also returns all graph pathfinding measurement patterns and so may still achieve teleportation even if at most two out of three central column (red) qubits are lost.
		Not only does the latter case provide additional qubit loss tolerance, but also tolerance to loss events that are only heralded at the point of measurement (i.e.\ unheralded loss).}
	\label{fig:both_multi_paths}
	\vspace{-5pt}
\end{figure}


However, not all stabilizers are equally useful for the task of producing teleportation measurement patterns.
To see this, consider applying stabilizer pathfinding for teleportation from $I$ to $O$ on linear graph state $\ket{\Psi}$ depicted below:
\begin{center}
\begin{varwidth}{0.35\linewidth}
	\begin{center}	
		\begin{figure}[H]
			\vspace{20pt}
		    \includegraphics[width=\textwidth]{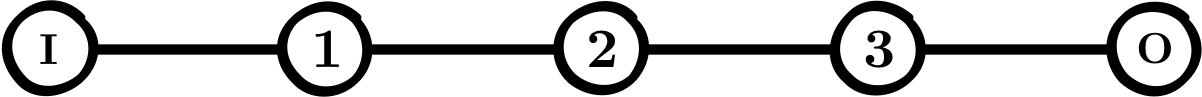}
		\end{figure}
	\end{center}
\end{varwidth}%
\hspace{0pt}
\begin{minipage}{0.45\linewidth}
	\vspace{0pt}
	\begin{align}	
		\arraycolsep=0.7pt
		\begin{array}{rlllllllc}
			\bar{X}_\Psi = &     & Z_I &     &     &     &     &    & \\
			\bar{Z}_\Psi = &     & X_I & Z_1 &     &     &     &    & \\
		\mathcal{G}_\Psi = & \{  & Z_I & X_1 & Z_2 &     &     & ,  & \; (K_1) \\ 
			   			   &     &     & Z_1 & X_2 & Z_3 &     & ,  & \; (K_2) \\ 
					       &     &     &     & Z_2 & X_3 & Z_O & ,  & \; (K_3) \\
	 	       			   &     &     &     &     & Z_3 & X_O & \} & \; (K_O)
		\end{array} \notag
	\end{align}
\end{minipage}
\end{center}

Firstly consider the stabilizer $S_{\{1,3\}} = K_1K_3 = Z_IX_1X_3Z_O$, used to define the logical operator $\bar{X}_{\{1,3\}} = S_{\{1,3\}} \bar{X}_\Psi = X_1X_3Z_O$.
This choice of stabilizer allows $\bar{X}_{\{1,3\}}$ to be paired with some $\bar{Z}$ that obeys the stabilizer pathfinding conditions for output qubit $O$.
Specifically, $\bar{Z}_{\{2,O\}} = S_{\{2,O\}} \bar{Z}_\Psi = X_I X_2 X_O$ satisfies equation \eqref{eq:motivation:log_op_pairs} with $M = \{X_I, X_1, X_2, X_3\}$, in this case reproducing the measurement pattern provided by graph-pathfinding.

Now consider the stabilizer $S_{\{1,O\}} = K_1 K_O = Z_I X_1 Z_2 Z_3 X_O$, used to define the logical operator $\bar{X}_{\{1,O\}} = S_{\{1,O\}} \bar{X}_\Psi = X_1 Z_2 Z_3 X_O$.
In this case $\bar{X}_{\{1,O\}}$ cannot be paired with any $\bar{Z}$ to satisfy equation \eqref{eq:motivation:log_op_pairs} to yield a valid measurement pattern.
This can be seen by observing that $\bar{X}_{\{1\}} = X_1 Z_2$ is also a valid $\bar{X}$ operator.
Hence, any measurement pattern constructed from $\bar{X}_{\{1,O\}}$ and some $\bar{Z}$ must contain measurements $X_1$ and $Z_2$ returning eigenvalues $\lambda_{X_1}$ and $\lambda_{Z_2}$ respectively.
However, $\eval{X_\Psi} = \eval{\bar{X}_{\{1\}}} = \lambda_{X_1}\lambda_{Z_2}$, showing that after such measurements $\bar{X}$ has been measured and thus teleportation has failed.


In this last example it is easy to see why $S_{\{1,O\}}$ cannot be used to generate an $\bar{X}$ satisfying equation \eqref{eq:motivation:log_op_pairs} by noting that $I \in \qsupp{K_1}$, $O \in \qsupp{K_O}$ but $\qsupp{K_1} \cap \qsupp{K_O} = \emptyset$, where $\qsupp{A}$ is the set of qubits on which $A$ non-trivially acts.
However it is not always the case that if some set of generators in a stabilizer combination share support then their combination is useful for stabilizer pathfinding
For example, consider applying stabilizer pathfinding to teleportation from $I$ to $O$ on star graph state $\ket{\Psi}$ depicted below:
\begin{center}
\vspace{-10pt}
\begin{varwidth}{0.35\linewidth}
	\begin{center}	
		\begin{figure}[H]
			\vspace{15pt}
		    \includegraphics[width=\textwidth]{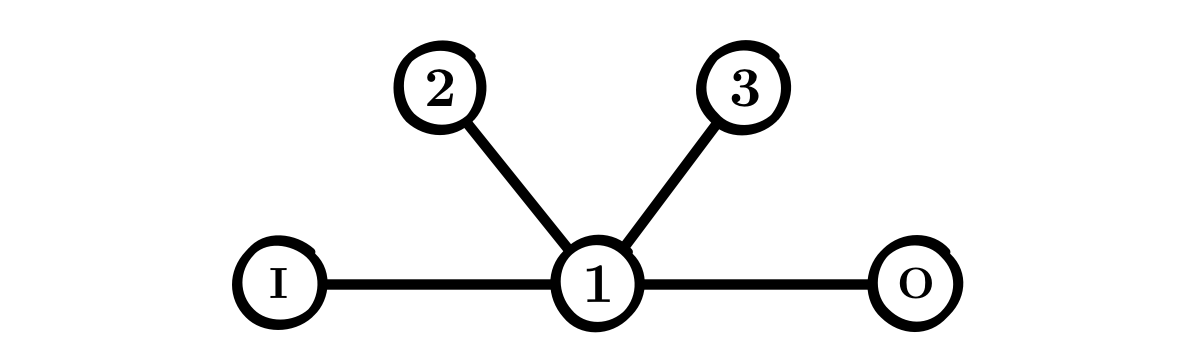}
		\end{figure}
	\end{center}
\end{varwidth}%
\hspace{-10pt}
\begin{minipage}{0.45\linewidth}
	\vspace{0pt}
	\begin{align}	
		\arraycolsep=0.7pt
		\begin{array}{rlllllllc}
			\bar{X}_\Psi = &     & Z_I &     &     &     &     &    & \\
			\bar{Z}_\Psi = &     & X_I & Z_1 &     &     &     &    & \\
		\mathcal{G}_\Psi = & \{  & Z_I & X_1 & Z_2 & Z_3 & Z_O & ,  & \; (K_1) \\ 
			   			   &     &     & Z_1 & X_2 &     &     & ,  & \; (K_2) \\ 
					       &     &     & Z_1 &     & X_3 &     & ,  & \; (K_3) \\
	 	       			   &     &     & Z_1 &     &     & X_O & \} & \; (K_O)
		\end{array} \notag
	\end{align}
\end{minipage}
\vspace{-12pt}
\end{center}

Consider the valid logical operator $\bar{Z}_{\{2,3,O\}} = S_{\{2,3,O\}}\bar{Z}_\Psi = X_I X_2 X_3 X_O$ on $\ket{\Psi}$.
Here we observe that $\qsupp{\bar{Z}_{\{O\}}} \cap \qsupp{K_2K_3} = \emptyset$ and therefore $\bar{Z}_{\{2,3,O\}}$ represents the same logical operation as $\bar{Z}_{\{O\}} = X_I X_O$ with $I, O \in \qsupp{\bar{Z}_{\{O\}}}$.
Even though in this case the inclusion of $K_2$ and $K_3$ does not prevent $\bar{Z}_{\{O\}}$ from acting on $I$ and $O$, $\bar{Z}_{\{2,3,O\}}$ still cannot be paired with any $\bar{X}$ that satisfies the stabilizer pathfinding condition.
This is seen by observing that any $\bar{X}$ must be produced using $K_1$ to ensure $\acomm{\bar{X}^{[O]}}{\bar{Z}^{[O]}} = 0$, and so qubits 2 and 3 must be measured in either the $Z$ or $Y$ basis.
On the other hand, a valid pair of logical operators satisfying equation \eqref{eq:motivation:log_op_pairs} would be $\bar{Z}_{\{O\}} = X_I X_O$ with $\bar{X}_{\{1\}} = X_1 Z_2 Z_3 Z_O$ such that $M = \{X_I, X_1, Z_2, Z_3\}$, also reproducing the measurement pattern provided by graph-pathfinding.

From the above examples we have illustrated that while many possible logical operators exist, only a subset can be used to produce valid measurement patterns.
Specifically, we have seen that teleportation can be prevented by logical operators which are decomposable into another logical operator (of reduced weight) and a non-overlapping stabilizer.
We now introduce definitions to generalise this concept and explicitly specify which stabilizers are useful for teleportation.

\subsubsection{Trivial and non-trivial stabilizers} \label{sec:SPF:relevant_stabs:triviality}


Given the correspondence between logical operators and stabilizers, we shall define general conditions on the latter.
To distinguish generator combinations that are and aren't useful for teleportation, we define the concepts of \emph{non-trivial} and \emph{trivial} combinations, respectively.
A trivial stabilizer (produced by a trivial combination) is defined as a stabilizer $S_c$ where there exists some bipartition $(\alpha, \beta)$ of $c$ such that the bipartition's stabilizers do not share support, or
\begin{align}
	 S_c = S_\alpha S_\beta \quad \textrm{where} \quad \qsupp{S_\alpha} \cap \qsupp{S_\beta} = \emptyset.
\end{align}

If, as in the examples above, a logical operator $\bar{L} \in \mathcal{L}_\Psi$ decomposes in a similar way\footnote
{
	Technically, valid logical operators can also be made from trivial stabilizers, however these are generally unhelpful for teleportation and can easily be allowed for when they arise..
	For further discussion, see Appendix \ref{apx:SPF:algorithm_p2}.
}
i.e.\ $\qsupp{\bar{L}^\prime S_\alpha} \cap \qsupp{S_\beta} = \emptyset$ or $\qsupp{\bar{L}^\prime S_\beta} \cap \qsupp{S_\alpha} = \emptyset$ for $\bar{L}^\prime \in \mathcal{L}_\Psi$, then qubits $I$ and $O$ must either both be in the support of just one of the partitions or split across both.
In such cases, $\bar{L}$ either has unnecessary measurements that can prevent teleportation, or measurements which simply do not help teleport the input state onto $O$.
The definitions of trivial and non-trivial logical operators are summarised in box \ref{box:t_nt_logical_ops}.

A non-trivial stabilizer (produced by a non-trivial combination) is conversely defined as a stabilizer $S_c$ for which no such bipartition of $c$ exists, or equivalently $\qsupp{S_\alpha} \cap \qsupp{S_\beta} \neq \emptyset$ for all possible bipartitions $(\alpha, \beta)$ of $c$.
Non-trivial stabilizers produce logical operators that can be used to teleport from $I$ to $O$ and do not contain unnecessary measurements.
For a given stabilizer state $\ket{\Psi}$ we denote the subsets of trivial and non-trivial stabilizers as $\mathcal{S}_\Psi^{\mathrm{T}}$ and $\mathcal{S}_\Psi^{\mathrm{NT}}$ respectively, such that $\mathcal{S}_\Psi = \mathcal{S}_\Psi^{\mathrm{T}} \cup \mathcal{S}_\Psi^{\mathrm{NT}}$.


The task of stabilizer pathfinding is therefore to track all of $\mathcal{S}^{\mathrm{NT}}_\Psi$ without explicit tracking of $\mathcal{S}^T_\Psi$ as $\ket{\Psi}$ is subject to gates, measurements and the addition of new qubits.
For each operation $\ket{\Psi} \mapsto \ket{\Psi^\prime}$, stabilizer pathfinding must therefore be able to add the set of stabilizers that are newly non-trivial $\mathcal{S}^{\mathrm{NT}}_{\Psi^\prime} \setminus \mathcal{S}^{\mathrm{NT}}_{\Psi}$, remove the set of newly trivial stabilizers $\mathcal{S}^{\mathrm{T}}_{\Psi^\prime} \setminus \mathcal{S}^{\mathrm{T}}_{\Psi}$, and apply an update to any non-trivial stabilizers that remain so.

\subsection{Tracking non-trivial stabilizers} \label{sec:SPF:tracking_nt_stabs}

To simulate the preparation of a quantum state using some Clifford circuit, our algorithm must simulate four operations: i) preparation of qubits in computational basis states $\{\ket{0}, \ket{1}\}$; ii) the single-qubit Clifford gates $H$, and $S$; iii) the two-qubit Clifford $\CZ$ gate; and iv) measurements in the computational basis.
We also require the algorithm to be described by some small set of update rules, whereby each successive operation is simulated by updating an internal representation of the state (as opposed to rerunning a complete simulation for each new state).
A simulation based on update rules is preferred not only for speed but also for practical purposes as it may be implemented in real-time.

\begin{box_fig}[t]
	\begin{framed}
		{\bf TRIVIAL AND NON-TRIVIAL LOGICAL OPERATORS:} \vspace{4pt}
		
		Consider the task finding pairs of logical operators $\mathcal{L}_\Psi$ that satisfy the stabilizer pathfinding conditions defined in box \ref{box:spf_conds} for teleportation from qubit $I$ to $O$ on $\ket{\Psi}$.
		A logical operator $\bar{L} \in \mathcal{L}_\Psi$ is known as \emph{trivial} if it can be decomposed into some other lower-weight logical operator $\bar{L}^\prime \in \mathcal{L}_\Psi$ and stabilizer $S \in \mathcal{S}_\Psi$ with non-overlapping qubit supports $\qsupp{\bar{L}^\prime} \cap \qsupp{S} = \emptyset$.
		For trivial $\bar{L}$, either
		\vspace{-2pt}
		\begin{enumerate}[label=\alph*)]
			\itemsep0em 
			\item $O \in \qsupp{S}$ and $\bar{L}$ cannot be used for teleportation, or
			\item $O \notin \qsupp{S}$, and $\bar{L}$ contains operators $S$ unnecessary for teleportation.
		\end{enumerate}
		Therefore trivial logical operators should not be considered for teleportation.
		By contrast, \emph{non-trivial} logical operators are those for which no such decomposition exists and so represent measurements that can produce teleportation.
	\end{framed}
	\vspace{-7pt}
	\caption{Definitions for trivial and non-trivial logical operators.}
	\label{box:t_nt_logical_ops}
\end{box_fig}

\subsubsection{Adding qubits and acting gates} \label{sec:SPF:tracking_nt_stabs:qubits_and_gates}


For appending a single qubit $\ket{\Psi^\prime} = \ket{\Psi}\otimes\ket{0}_{n+1}$, the state generators acquire one additional non-trivial stabilizer $\mathcal{G}_{\Psi^\prime} = \mathcal{G}_{\Psi} \cup \{Z_{n+1}\}$ and so $\mathcal{S}^{\mathrm{NT}}_{\Psi^\prime} = \mathcal{S}^{\mathrm{NT}}_{\Psi} \cup \{Z_{n+1}\}$ is updated accordingly.

For the case of applying the single-qubit Clifford gate $U$
\begin{align}
	\mathcal{S}_{\Psi^\prime} = \{U S_c U^\dagger \;\;\forall\;\;\ S_c \in \mathcal{S}_\Psi\} \label{eq:SPF:stab_update}
\end{align}
In Remark \ref{apx:lems:single_q_unitaries} we also show that the action $U$ cannot affect the non-triviality of any stabilizer.


For the two-qubit $\CZ$ gate, finding $S_c^\prime \in \mathcal{S}^{\mathrm{NT}}_{\Psi^\prime}$ is more involved as new non-trivial and trivial stabilizers may be generated.
We provide an example here with the update rule's full description found in Appendix section \ref{apx:SPF:algorithm_p1:cz_gates}.
Consider the following graph state produced by applying $\CZ_{6,8}$ to the state depicted in figure \ref{fig:three_ways_1}:
\begin{figure}[H]	
\begin{center}
	\begin{varwidth}{0.35\linewidth}
		\begin{center}	
			\begin{figure}[H]
			    \includegraphics[width=\textwidth]{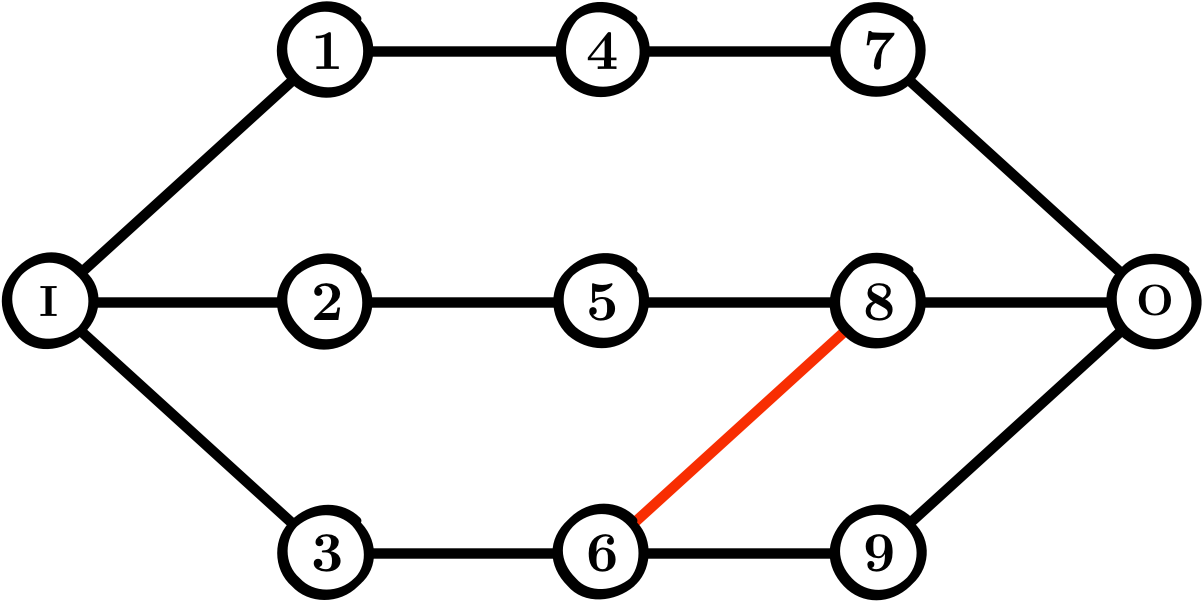}
			\end{figure}
		\end{center}
	\end{varwidth}%
	\begin{minipage}{0.45\linewidth}
		\vspace{-10pt}
		\begin{align}	
			\arraycolsep=1.4pt
			\begin{array}{rlll}
				\bar{X} = &\;\, Z_I \\
				\bar{Z} = &\;\, X_IZ_1Z_2Z_3 \\
				\mathcal{G}_{\psi} = &\{Z_IX_1Z_4, & Z_IX_2Z_5, & Z_IX_3Z_6, \\
										  	  &\;\, Z_1X_4Z_7, & Z_2X_5Z_8, & \red{Z_3X_6Z_8Z_9}, \\
											  &\;\, Z_4X_7Z_O, & \red{Z_5Z_6X_8Z_O}, & Z_6X_9Z_O, \\
											  &\multicolumn{2}{l}{\;\,Z_7Z_8Z_9X_O \}}
			\end{array} \notag
		\end{align}
	\end{minipage}
	\vspace{-22pt}
\end{center}
\caption{The graph state produced by applying $\CZ_{6,8}$ to the state depicted in figure \ref{fig:three_ways_1}.}
\label{fig:three_ways_2}
\vspace{-5pt}
\end{figure}
\noindent where the action of $\CZ_{6,8}$ has been highlighted and the generators are indexed as before (by the qubit on which the Pauli $X$ operator acts).


From inspection, it is seen that many stabilizers' triviality are unchanged, for example $S_{\{5,9\}}^\prime = Z_2 X_5 Z_6 Z_8 \penalty 0 X_9 Z_O$ and $S_{\{6,O\}}^\prime = Z_3 X_6 Z_7 X_O$, remain trivial and non-trivial respectively.
On the other hand, we see that $S_{\{6,8\}}^\prime = Z_3 Z_5 Y_6 Y_8 Z_9 Z_O \in \mathcal{S}^{\mathrm{NT}}_{\Psi^\prime}$, whereas $S_{\{6,8\}} = Z_3 Z_5 X_6 X_8 Z_9 Z_O \in \mathcal{S}^{\mathrm{T}}_{\Psi}$ under bipartition $(\{6\},\{8\})$.
We can also find examples of newly trivial stabilizers, for example $S_{\{5,6,O\}}^\prime = Z_2 Z_3 X_5 X_6 Z_7 Z_8 X_0 \in \mathcal{S}^{\mathrm{T}}_{\Psi^\prime}$ under bipartition $(\{5\},\{6, O\})$, whereas $S_{\{5,6,O\}} = Z_2 Z_3 X_5 X_6 Z_7 X_0 \in \mathcal{S}^{\mathrm{NT}}_{\Psi}$.

Although small, low-connectivity graph states are easy to analyse, larger graph states or non-graphical stabilizer states become increasingly difficult with a rapidly growing number of combinations available.
Our approach identifies new trivial and non-trivial stabilizers using only information of the stabilizers in $\mathcal{S}^{\mathrm{NT}}_{\Psi}$.
Since there are $2^{\abs{c}}$ possible bipartitions of any given $S_c$, when a test of triviality is needed, our method avoids an exhaustive search by identifying a reduced set of bipartitions to be tested.
Once all stabilizers with differing triviality have been found, the remaining non-trivial stabilizers can then be simply updated as described by equation \eqref{eq:SPF:stab_update}.

\subsubsection{Single-qubit measurements} \label{sec:SPF:tracking_nt_stabs:single_qubit_mnts}


Finally, we consider performing single-qubit Pauli measurements on the state.
As with the $\CZ$ gate, Pauli measurements may also affect the triviality of a given stabilizer.
For example, consider the state produced by measurement of $Y_9$ (followed by applying corrective gates $S^3$ on qubits 6 and $O$) on the previous state, as depicted below:
\begin{figure}[H]
\begin{center}
\begin{varwidth}{0.35\linewidth}
	\begin{center}	
		\begin{figure}[H]
		    \includegraphics[width=\textwidth]{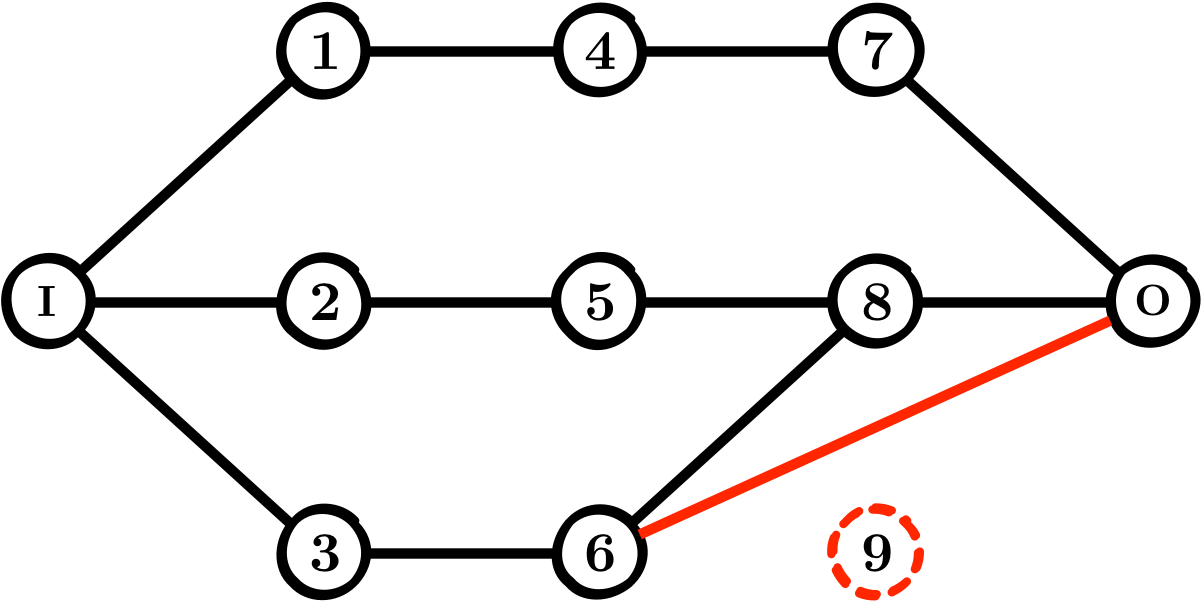}
		\end{figure}
	\end{center}
\end{varwidth}%
\begin{minipage}{0.45\linewidth}
	\vspace{-10pt}
	\begin{align}	
		\arraycolsep=1.4pt
		\begin{array}{rlll}
			\bar{X} = &\;\, Z_I & &\\
			\bar{Z} = &\;\, X_IZ_1Z_2Z_3 & & \\
				\mathcal{G}_{\psi} =		&\{   Z_IX_1Z_4, & Z_IX_2Z_5, & Z_IX_3Z_6, \\
										&\;\, Z_1X_4Z_7, & Z_2X_5Z_8, & \red{Z_3X_6Z_8Z_O}, \\
										&\;\, Z_4X_7Z_O, & \red{Z_5Z_6X_8Z_O}, & \red{Y_9}, \\
								       	&\multicolumn{2}{l}{\;\, Z_6Z_7Z_8X_O \}}
		\end{array} \notag
	\end{align}
\end{minipage}
\vspace{-22pt}
\end{center}
\caption{The graph state produced by measurement of $Y_9$ on the state depicted in figure \ref{fig:three_ways_2}.}
\label{fig:three_ways_3}
\vspace{-5pt}
\end{figure}
\noindent where the measurement's action has been highlighted and we have assumed qubit 9 is found in the $+1$ $Y$ eigenstate.

Again we see that many stabilizers' triviality are unchanged, such as $S_{\{5,9\}}^\prime = Z_2 X_5 Z_8 Y_9$ and $S_{\{6,O\}}^\prime = - Z_3 X_6 Z_7 X_8 X_O $ as before.
Similarly, new non-trivial stabilizers can be found, for example $S_{\{3,O\}}^\prime = Z_I X_3 Z_7 Z_8 X_O \in \mathcal{S}^{\mathrm{NT}}_{\Psi^\prime}$, whereas before $S_{\{3,O\}} = Z_I X_3 Z_6 Z_7 Z_8 Z_9 X_O \in \mathcal{S}^{\mathrm{T}}_{\Psi}$ under bipartition $(\{3\},\{O\})$.
Lastly, we also find new trivial stabilizers, for example $S_{\{6,7,8\}}^\prime = Z_3 Z_4 Z_5 Y_6 X_7 Y_8 Z_O \in \mathcal{S}^{\mathrm{T}}_{\Psi^\prime}$ under bipartition $(\{6, 8\},\{7\})$, whereas prior to measurement $S_{\{6,7,8\}} = Z_3 Z_4 Z_5 Y_6 X_7 Y_8 Z_9 \in \mathcal{S}^{\mathrm{NT}}_{\Psi}$.
As before, identifying the full set of stabilizers with triviality changed by measurement is somewhat involved, however our algorithm does achieves this with knowledge of only $\mathcal{S}^{\mathrm{NT}}_{\Psi}$ and without the need for exhaustive triviality testing.


It must be noted that while we could not find an analytic expression for the worst-case efficiency of our algorithm, it will be highly state-specific and more crucially depend on intermediate states produced during the state's construction.
These rules are therefore most efficient for states at or close to their \emph{minimal edge representation} (or equivalent for non-graph states) \cite{Hein2004}.
For example, while for a completely connected graph state of $n$ qubits $S_c \in \mathcal{S}^{\mathrm{T}}_{\Psi} \;\forall\; \abs{c} \geq 4$, $c$ even, $2^n$ intermediate states must also be constructed and clearly such a construction would be inefficient.
In these cases alternative construction strategies should be considered.
For example, the previous state can be more efficiently created by first creating a $n+1$ star graph state (which is a minimal edge representation of the $n+1$ completely-connected graph state), followed by the measurement of the central qubit in the $Y$ basis.
While optimal construction strategies are beyond the scope of this paper, we note that minimum edge representation states are likely to be of interest for MBQC in many scenarios 
For a further discussion of ways to increasing the algorithm's efficiency, see section \ref{sec:discussion:opts_and_exts}.

\subsection{Finding loss-tolerant measurement patterns} \label{sec:SPF:finding_mnt_pats}


Once $\mathcal{S}^{\mathrm{NT}}_{\Psi}$ are known, the set of all non-trivial logical operators $\mathcal{L}^\mathrm{\mathrm{NT}}_\Psi$ and valid measurement patterns $\mathcal{M}$ can be found.
Our algorithm is designed to produce those $M$ which can tolerate the most lost first, so that only a fraction of all $\mathcal{M}$ need be found.
This is achieved by grouping $\mathcal{L}^\mathrm{\mathrm{NT}}_\Psi$ into three\footnote
{
	Here all three Pauli operators must be considered (rather than just $X$ and $Z$) because although all $\bar{Y}$ operators may be produce by a product of some $\bar{X}$ and $\bar{Z}$, it is possible that $\bar{Y}$ acts non-trivially on fewer qubits than both $\bar{X}$ and $\bar{Z}$.
}
 sets defined by the Pauli operator on qubit $O$, namely $X_O$, $Y_O$ and $Z_O$.
Within each group operators are then further sorted into groups of equal weight.
All minimum-weight $M$ are then be found by considering pairings of operators taken from the lowest-weight operators in groups where $\acomm{A_O}{B_O}=0$.
Higher weight $M$ can then be iteratively produced by considering pairing between lowest-weight and second-to-lowest-weight groupings, etc.


Once some subset of $\mathcal{M}$ is known, each $M$ provides some set of loss-tolerant qubits and hence the set of all qubit loss configurations can be easily found.
In practise we find that the majority of loss tolerance is provided by a few low-weight $M$ that are among the first to be found---see numerical results provided in Appendix \ref{apx:results:mnt_strategies}.
For a more detailed description of the above algorithm, see Appendix \ref{apx:SPF:algorithm_p2}.

\section{Loss tolerance} \label{sec:loss_tol}

\begin{figure}[t]
	\centering
	\vspace{0pt}
	\includegraphics[width=\textwidth]{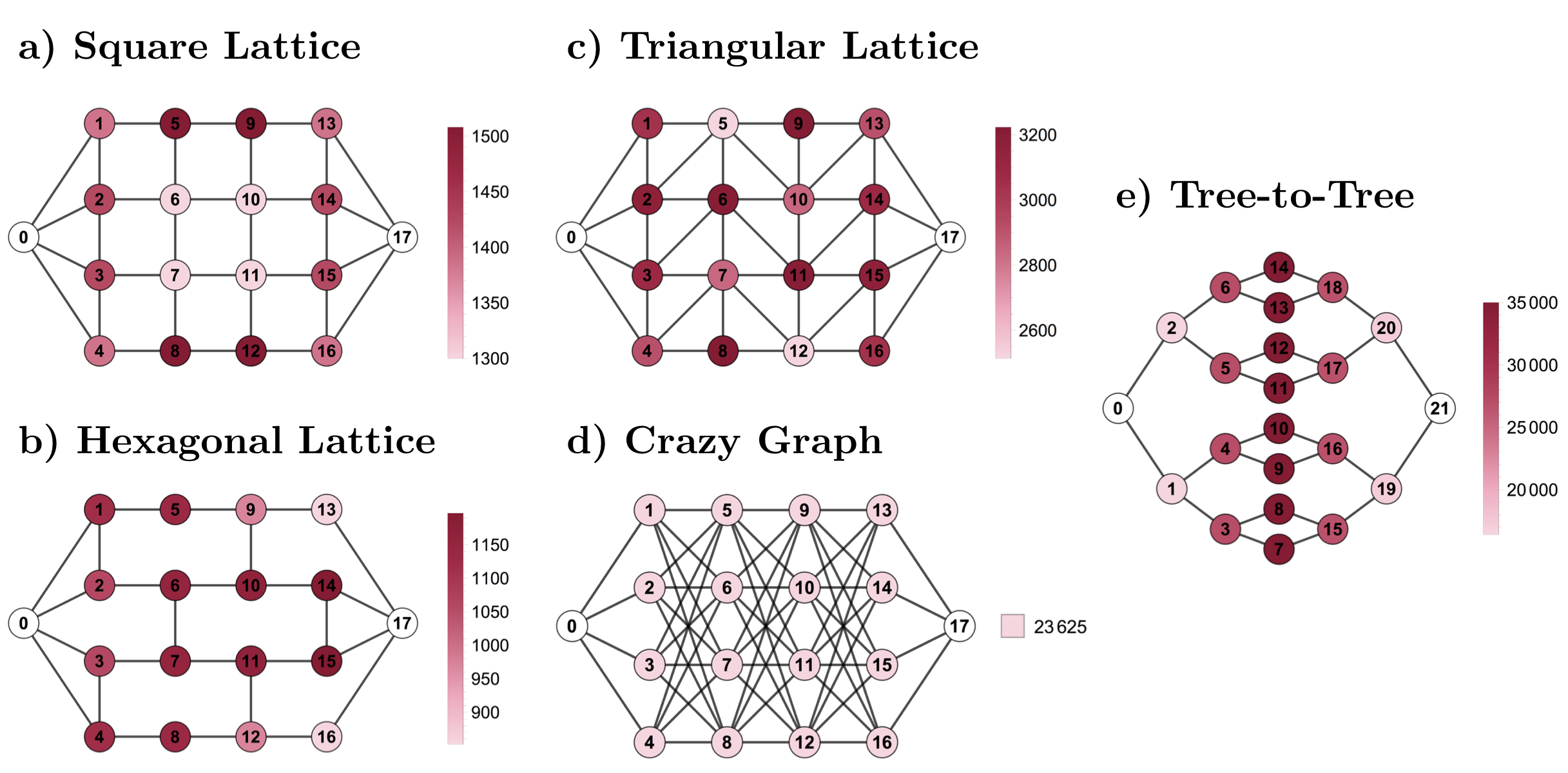}
	\caption{
		The five graph state channels considered for teleportation.
		Nodes are coloured by the number of measurement patterns that do not contain them (i.e.\ tolerant to their loss), with darker nodes indicating their loss can be more readily accommodated.}
	\label{fig:all_channels}
	\vspace{-5pt}
\end{figure}


To assess the loss tolerance of stabilizer pathfinding we compare the performance of GPF and SPF on a selection of graph state channels.
Specifically, we consider the five channels depicted in figure \ref{fig:all_channels}: the square lattice, hexagonal lattice, triangular lattice, linear \emph{crazy graph}, and a \emph{tree-to-tree} graph.
The choice of three lattice channels is motivated by their relevance to MBQC architectures; the so-called crazy graph is considered due it's use as a loss-tolerant qubit channel \cite{Rudolph2016} and a tree-to-tree channel because it supports a high number of disjoint paths.


We consider two kinds of loss: \emph{heralded} and \emph{unheralded}.
Heralded refers to loss events whose location is known, whereas unheralded to loss events on qubits whose locations are unknown until measurement.
Physically, heralded loss occurs when a qubit's existence can be inferred from some non-destructive measurement; for example, measurement of charge in a quantum dot can herald the existence of a spin-encoded qubit without measuring the qubit state.
On the other hand, unheralded loss occurs in qubit systems that do not permit such measurements, such as a dual-rail encoded qubit in linear optics where measurements are typically performed using photon detectors which absorb the photons (such as avalanche photodiodes).


Importantly, unheralded loss presents a significant challenge to any MBQC scheme as it necessitates either loss-tolerantly encoded qubits or an architecture that can adapt dynamically to loss events when they occur.
However, even in a system with unheralded loss, the performance of SPF under heralded loss provides an upper bound on the loss tolerance of any given channel or teleportation measurement strategy.

\subsection{Heralded loss} \label{sec:loss_tol:heralded}

\begin{figure}[t]
	\centering
	\vspace{0pt}
	\includegraphics[width=0.7\textwidth]{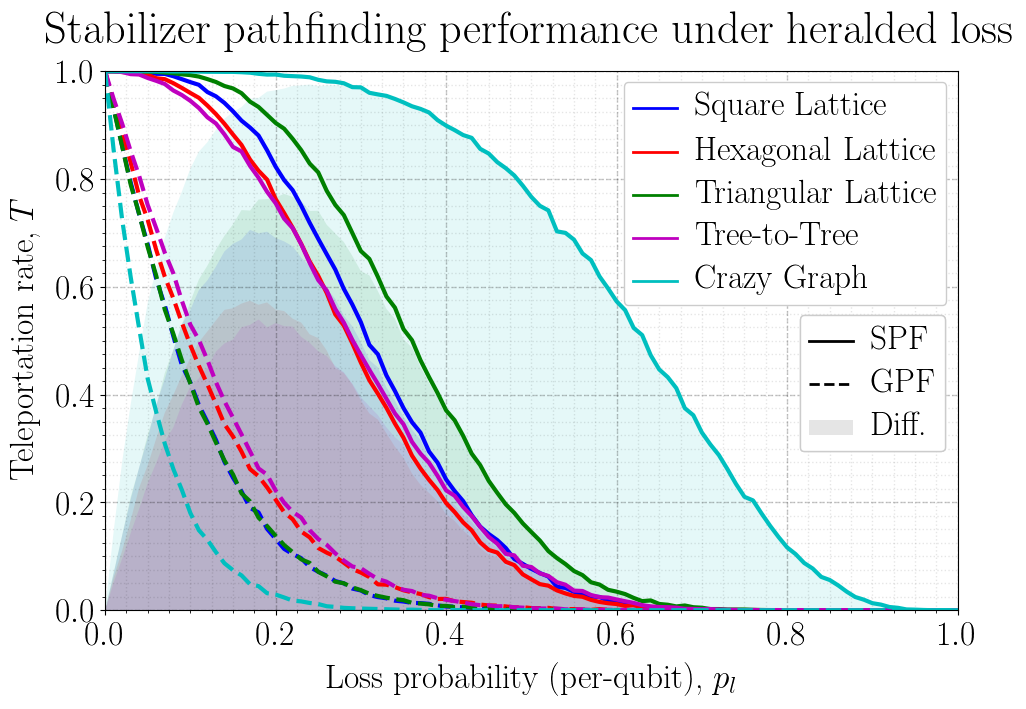}
	\caption{
		The probability of successful teleportation across various graph state channels as a function of heralded qubit loss.
		Solid lines, dashed lines and shaded regions depict the performance of stabilizer pathfinding (SPF), the graph pathfinding (GPF) heuristic and the difference between them respectively (i.e.\ SPF's loss tolerance advantage).
		Each data point depicts $10^4$ Monte Carlo instances and uncertainties have not been plotted as $\Delta T < 0.5\% $ in all cases and so are smaller than the plotted lines.
		From these results it is clear that stabilizer pathfinding produces a significant improvement in terms of the loss tolerance of teleportation across these states.
		Additionally, the loss tolerance provided by SPF for the crazy graph channel agree with the theoretical prediction of $T=(1-p_l^m)^n$ for the case of $n, m = 4$ presented here.}
	\label{fig:spf_vs_gpf_tel_rates}
	\vspace{-5pt}
\end{figure}


Firstly, we consider the performance of SPF in the case of heralded loss.
Once a set of measurement patterns for a state is known (be they produced by GPF or SPF), the rate of successful teleportation as a function of per-qubit loss probability can be easily found by Monte Carlo simulation.
Specifically, for a single Monte Carlo instance this is achieved by randomly generating some set of lost qubits (at some per-qubit loss rate $p_l$), which is cross-referenced with the set of all measurement patterns to find any pattern that do no include said qubits to allow successful teleportation.
In figure \ref{fig:spf_vs_gpf_tel_rates} we compare the performance of SPF to that of GPF on the five aforementioned channels under heralded loss.


Firstly, it is clear that SPF provides a significant increase in the loss tolerance of teleportation rate $T$.
As should be expected, GPF has greatest loss tolerance on the tree-to-tree channel and lowest on the crazy graph (where it can tolerate no loss) whereas the converse is true for SPF respectively.
For all channels considered, the SPF's gain in loss tolerance peaks above $50\%$, even for the tree-to-tree channel.
Note that the SPF teleportation rate for crazy graph agrees with the theoretical rate\footnote
{
	Specifically, teleportation succeeds if at least one qubit per column is measured in the $X$ basis, allowing all but one physical qubit to be lost per encoded qubit.
} of $T=(1-p_l^m)^n$, where $m$ and $n$ are the number of qubits per column and channel depth respectively (with $m = n = 4$ in the case considered).
We further note that in the low-loss regime for $p_l < 10\%$ SPF achieves $T > 95\%$ for all lattice channels and even $T \approx 1$ for the triangular lattice.

\subsection{Unheralded loss} \label{sec:loss_tol:unheralded}

\begin{figure}[t]
	\centering
	\vspace{0pt}
	\includegraphics[width=0.7\textwidth]{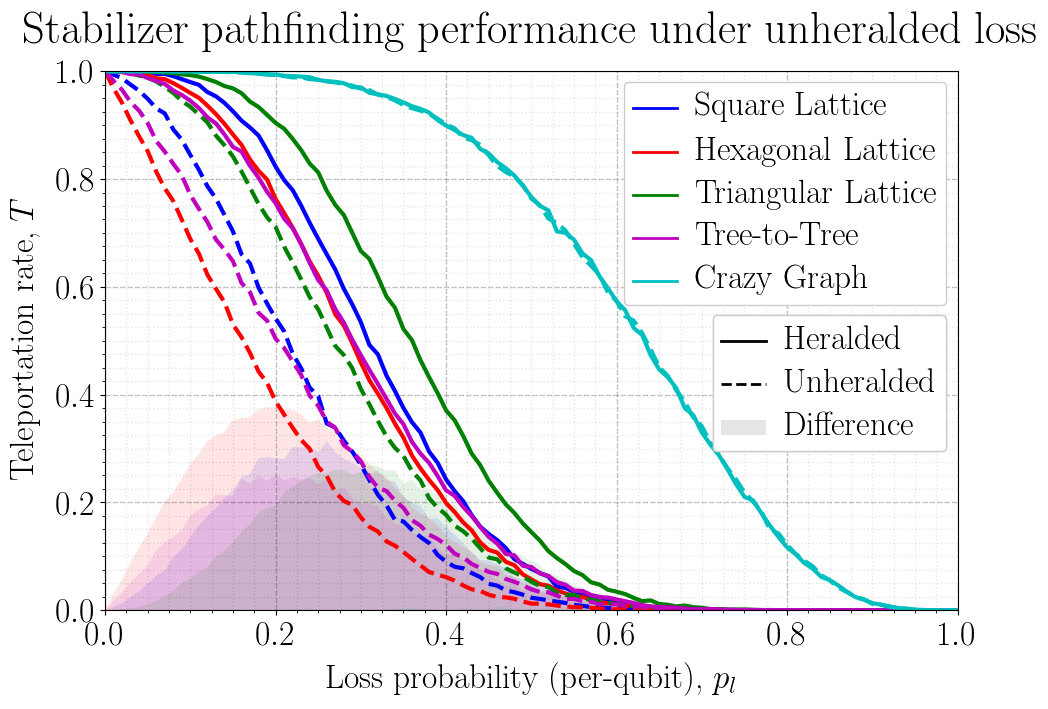}
	\caption{
		The probability of successful teleportation across various graph state channels as a function of qubit loss in the heralded and unheralded case.
		Solid lines, dashed lines and shaded regions depict the performance of stabilizer pathfinding with heralded loss, unheralded loss and the difference between them respectively (i.e.\ the decrease in loss tolerance due to unheralded loss).
		Each data point depicts $10^4$ Monte Carlo instances and uncertainties are not depicted as $\Delta T < 0.5\% $ in all cases and so are smaller than the plotted lines.
		From these results it is clear that while unheralded loss does reduce the ability to teleport loss-tolerantly on MBQC resource states, high teleportation rates can still be achieved, especially for those lattices with higher degree.
		We also note that as expected the crazy graph channel does not show any decrease in loss tolerance in the unheralded case, indicating its unique construction as a loss-tolerant teleportation channel.}
	\label{fig:uh_vs_h_tel_rates}
	\vspace{-5pt}
\end{figure}


We now consider the performance of SPF in the case of unheralded loss.
Any practical device that must tolerate unheralded loss during teleportation (without a loss-tolerant encoding) must be able to react to loss events as and when they occur.
One method for achieving this is to pre-compile a set of possible measurement patterns $\mathcal{M}$, many of which will contain common measurements.
Since teleportation can be achieved as long as one valid pattern can be performed that excludes all lost qubits, we require some \emph{measurement strategy} that finds at least one such measurement pattern with high probability.
For demonstrative purposes we only consider a single measurement strategy here, known as \emph{max tolerance}.
In the max tolerance measurement strategy the measurement that occurs most in the set of maximally loss-tolerant patterns is chosen; this process is then repeated until either a valid measurement pattern is completed and teleportation succeeds or none remain and teleportation fails.
For further details on this strategy and other considered see Appendix section \ref{apx:results:mnt_strategies}.

Specifically, at each Monte Carlo simulation instance a set of lost qubits is again generated and qubits are sequentially measured.
At each measurement, if the qubit is not lost then the measurement succeeds, and all measurement patterns not containing the measurement are discarded.
Conversely, if the measured qubit is lost, all measurement patterns that required measurement of the qubit are discarded.
Successful teleportation occurs when a successful measurement completes a measurement pattern, whereas if no patterns remain then teleportation fails.


Our Monte Carlo simulation results, depicted in figure \ref{fig:uh_vs_h_tel_rates}, indicates that teleportation is surprisingly resilient to unheralded loss across the channels considered.
Immediately, it is clear that the crazy graph channel does not experience any decrease in teleportation rate in the unheralded case.
This can be understood by noting that, unlike the other channels considered, the crazy graph is a loss-tolerant encoding of a four qubit linear graph state and is specifically designed to tolerate unheralded loss.

For lattice channels, the disadvantage of unheralded loss decreases with increased node degree.
Most importantly, the decrease of $T$ with unheralded loss is far more favourable for higher degree in the regime of low loss $p_l < 5\%$, with the triangular lattice showing only a $2\%$ decrease in loss tolerance.
We finally note that although a fall in $T$ is observed for unheralded loss, this drop is not as large as might be expected.
Most notably, SPF teleportation on the triangular lattice under unheralded loss performs almost as well as SPF teleportation on the hexagonal lattice under heralded loss (which already marks a significant improvement when compared to teleportation using GPF).
Overall, these results present an optimistic outlook on the future of designing loss-tolerant architectures for quantum computation and other quantum technologies based on such states.

\subsection{Loss tolerance thresholds} \label{sec:loss_tol:thresholds}

\begin{figure}[t]
	\centering
	\vspace{0pt}
	\includegraphics[width=\textwidth]{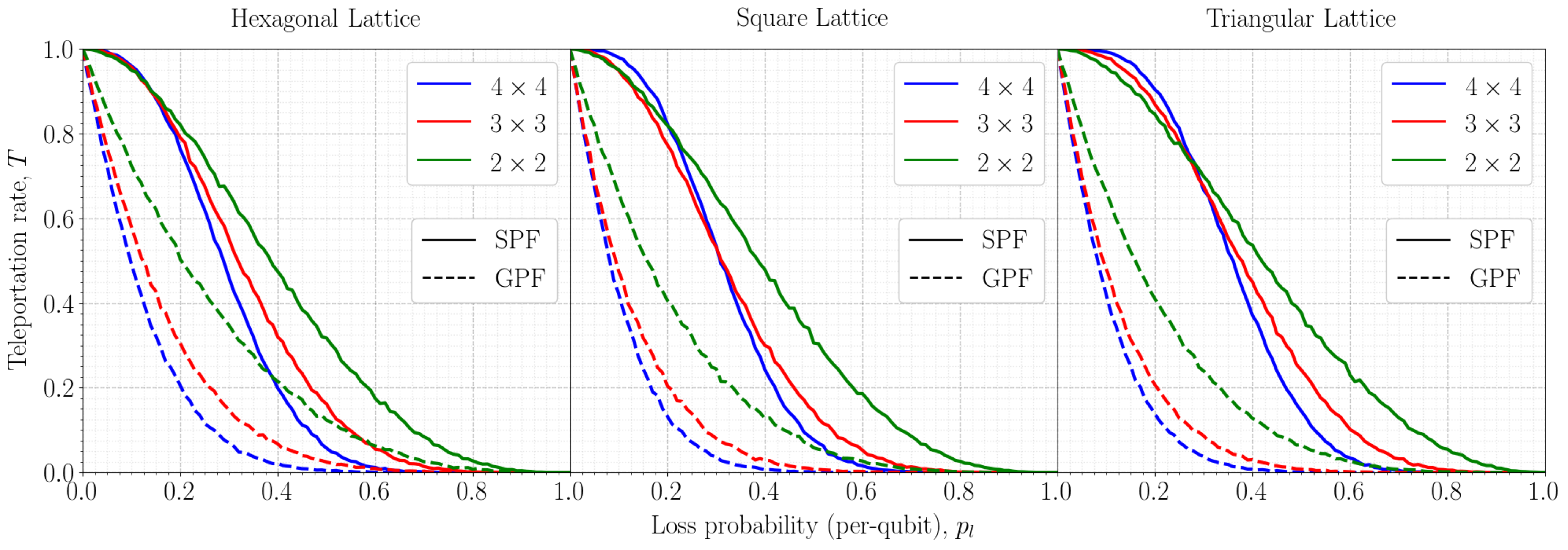}
	\vspace{-15pt}
	\caption{
		Loss-tolerant teleportation rates $T(p_l)$ for the three lattice channels considered: hexagonal, square and triangular over lattice sizes $2 \times 2$, $3 \times 3$ and $4 \times 4$ in the case of heralded loss.
		For the hexagonal and triangular lattice channels, stationary points appear at $p_l \approx 12$--$15\%$ and $p_l \approx 27$--$30\%$ respectively, whereas with the square lattice, no single crossing point occurs.
		If these points can be shown to represent a critical threshold in loss tolerance, then they define the channels' ultimate teleportation loss tolerance in limit of infinite channel size.
		However, $T(p_l)$ at higher $n$ are needed to verify the stationary points to prove such a conjecture.
	}
	\label{fig:GPF_vs_SPF_channel_thresholds}
	\vspace{-5pt}
\end{figure}


One interesting feature of loss tolerance provided by SPF is that $T(p_l)$ appears to exhibits threshold behaviour on the lattice channels considered here.
We conjecture that such a threshold does exist in the infinite limit, allowing a loss-tolerant threshold $p_l^*$ to be defined on these states.
If this conjecture holds, $p_l^*$ represents a distinct division in loss tolerance in the limit of infinite channel size (where $n \rightarrow \infty$ on an $n \times n$ lattice), where for $p_l<p_l^*$ loss-tolerant teleportation can always be achieved, whereas for $p_l>p_l^*$ it cannot.
It is known from percolation theory that the probability of finding a spanning path $\Gamma$ across some percolated lattice with edge/node percolation rate $p$ exhibits a threshold at some \emph{critical} probability $p=p^*$, which can be found from the stationary point in $\Gamma(p)$ between finite lattices of various sizes \cite{Stauffer1994}.

Figure \ref{fig:GPF_vs_SPF_channel_thresholds} depicts $T(p_l)$ for each lattice channel across lattice sizes $2 \times 2$, $3 \times 3$ and $4 \times 4$, with stationary points found for hexagonal and triangular lattice channels at $p_l \approx 12$--$15\%$ and $p_l \approx 27$--$30\%$ respectively.
No clear crossing is observed in the square lattice case.
While these results are not conclusive, it is surprising that any crossing points are found given the small lattice sizes considered here as larger systems are usually needed to overcome perturbative boundary effects.
In the square lattice case, we conjecture that no crossing occurs because of such boundary effects, given their effect on per-node loss tolerance as depicted in figure \ref{fig:all_channels}.
We also note that while the crazy graph lattice does also appear to show threshold behaviour (from the sigmoidal form of it's $T$ curve in figure \ref{fig:spf_vs_gpf_tel_rates}) for some threshold $p_l<1$ for the small sizes considered, this is not the case.
This is because as $n \rightarrow \infty$ for a $n \times n$ channel, $T = \lim_{n\rightarrow\infty}\left((1-p_l^n)^n\right) = 1$ for $p_l < 1$, in which case $p_l^* = 1$.

\section{Discussion} \label{sec:discussion}

We now provide possible optimisations and extensions of SPF and our algorithm as well as a discussion on it's applicability to various quantum architectures.

\subsection{Optimisations and extensions} \label{sec:discussion:opts_and_exts}


Our SPF algorithm allows for optimisation in various situations.
Because we have focused on achieving loss tolerance across all possible states, our implementation of SPF necessarily tracked all non-trivial stabilizers, making simulation of over 20 qubit graph states infeasible without significant computational power.
However, a more specialised implementation might suffice when applying SPF on a single type of channel, such as one which ignores isomorphic stabilizers on states with high symmetry.


More generally, many non-trivial stabilizers tracked during SPF have high weight and so typically don't contribute to loss-tolerant measurement patterns or to later future low-weight stabilizers as the state is grown.
Because only stabilizers of up to $\floor{\frac{n}{2}}$ generators are needed for triviality tests, all high-weight stabilizers produced from combinations of over $\floor{\frac{n}{2}}$ generators may be disregarded.
For large states this can significantly reduce computational runtime without either causing failure of our SPF algorithm or an appreciable reduction in loss tolerance.


Another route for optimisation and/or extension of SPF is provided by pre-compilation.
In a quantum architecture with probabilistic entangling operations within a fixed network structure, the non-trivial stabilizers for the ideal network may be pre-compiled (as an expensive but one-off computation) so that our algorithm can \emph{build down} (rather than build up) to the target state.
Alternatively, for large, regular graph state lattices, teleportation might be split up into many smaller SPF instances that are concatenated to produce the required long-range measurement patterns.
The challenge here is to ensure consistency across the boundaries between different SPF instances.

Finally, we observe that SPF can be extended to include parity checks for the detection of computational measurement errors.
This is seen by noting that each of a state's stabilizers provides a parity check for measurement on the state.
Therefore, if a stabilizer can be found whose non-trivial Pauli operators are a subset of the teleportation measurement pattern (or which contains additional available measurements), the operator provides a parity check on measurement outcomes.
Combinations of parity checks which overlap on sets of qubits can thus be used to detect Pauli measurement errors, as demonstrated by tolerance of up to a $50\%$ Pauli $Z$ error rate on crazy-graph states argued in Ref.\ \cite{Rudolph2016}.

\subsection{Relevance to quantum architectures} \label{sec:discussion:architectures}

Firstly, our results provide important progress towards addressing the problem of photon loss within linear optical quantum computer.
For example, some LOQC architecture proposals \cite{Kieling2007,Gimeno-Segovia2015,Zaidi2015,Morley-Short2018} overcome probabilistic entangling gates by the \emph{renormalization} of large blocks of percolated graph state to construct 3D topological error correction states such as the Raussendorf lattice \cite{Raussendorf2006a}.
But due their use of GPF for teleportation, these models previously lacked any tolerance to unheralded loss.
The loss tolerance thresholds conjectured in section \ref{sec:loss_tol:thresholds} indicate that loss tolerance can be straightforwardly achieved in these schemes by replacing GPF with SPF.

More generally, SPF can provide additional loss tolerance within many other quantum architectures without modification, before or after error-correction.
For example, given that flow conditions are unaffected by Pauli $X$, $Y$ and $Z$ measurements \cite{Browne2007}, SPF is readily compatible with teleportation in an MBQC architecture.
Similarly, because SPF makes no assumptions on the physical encoding of qubits, our work equally extends to teleportation of logical qubits which are encoded for quantum error correction or other reasons.
Hence in some systems it may be possible to substitute the resources associated with producing asymptotically-lossless logical qubits with the creation of a larger network of (heralded) low-loss logical qubits on which SPF can be applied.

One further aspect we have not explicitly addressed here is the ability to perform measurement-based qubit gates on top of an SPF teleportation scheme.
Unlike GPF, because no linear cluster is directly generated during SPF the standard MBQC gate protocol cannot be directly applied.
However, given that at least one unbroken qubit path of $Y$ and $X$ measurements must connect the input and output qubits, with all others effectively applying the necessary $Z$ measurements, it is straightforward to understand how standard MBQC protocols may be similarly implemented.
We leave a full description of such a protocol open for future works.

Lastly, it is clear from the results of section \ref{sec:loss_tol} that higher-degree graph states seem to provide a greater degree of loss tolerance in both the heralded and unheralded case.
As such, it is an open question whether this result generalises for arbitrary $n$-degree random graphs or lattices.
The identification of such a dependance would provide an important insight into the design of future network architectures.

\section{Conclusion and outlook} \label{sec:conclusion}

Qubit loss presents a substantial roadblock to the realistic implementation of teleportation within many large-scale quantum technologies, such as LOQC and quantum communication networks.
Previously, this could only be generally addressed through costly loss-tolerant encodings, especially so when qubit loss is unheralded.
However, by applying a generalised approach to teleportation, SPF, our work provides loss-tolerant teleportation on any stabilizer state using only single-qubit Pauli measurements and feed-forward.
We have show that SPF provides all maximally loss-tolerant teleportation measurement patterns (when loss is heralded) without use of an exhaustive search.
Furthermore we have shown that SPF also allows for significant degrees of unheralded qubit loss to be tolerated by dynamic and computationally-inexpensive measurement strategies.

In addition to theoretical analysis, we have provided an algorithm that implements SPF as well as unheralded measurement strategies which incur minimal computational cost.
Based on numerical simulations of SPF, we have further conjectured the existence of loss-tolerant thresholds on a variety of graph state lattices that exist in the limit of infinite lattice size.
From a practical perspective our results provide both a novel technique for tolerating loss in large-scale quantum architectures as well as a tool for maximal use of so-called noisy intermediate-scale quantum (NISQ) devices in the near future \cite{Preskill2018}.

\printbibliography

\section*{Acknowledgements}

This work was supported by the UK Engineering and Physical Sciences Research Council (EPSRC).
SMS is supported by the Bristol Quantum Engineering Centre for Doctoral Training, EPSRC grant EP/L015730/1.
SMS would like to thank Eric Johnston, Will McCutcheon, Stasja Stanisic, Sam Pallister, Chris Sparrow and Naomi Nickerson for fruitful discussions.
Finally, we thank two anonymous referees for their constructive comments and feedback.
All simulation code, data and analysis scripts are available for download and use at \url{https://github.com/sammorley-short/spf}.

\newpage

\begin{appendices} 

\let\clearpage\relax
\noindent{\huge\bfseries Appendices\par}~

\noindent The appendix is structured as follows.
Appendix \ref{apx:stab_theory} gives an introduction to the stabilizer formalism and stabilizer states.
Appendix \ref{apx:SPF} provides a complete description of the SPF algorithm, including necessary proofs and pseudocode.
Finally, Appendix \ref{apx:results} presents a number of further results and discussion of the SPF algorithm.

\section{Stabilizer states} \label{apx:stab_theory}

Below we review the necessary theoretical results on which our work relies, namely the stabilizer formalism, graph states, and a generalised teleportation.


In the stabilizer formalism \cite{Gottesman1997}, for any given state $\ket{\Psi}$ there exists an associated stabilizer group $\mathcal{S}_\Psi$, consisting of the set of all operators that leave $\ket{\Psi}$ unchanged, such that
\begin{equation}
	\mathcal{S}_\Psi = \{S_i : S_i\ket{\Psi} = \ket{\Psi}\}.
\end{equation}
A state's stabilizer group is closed under multiplication, i.e.\ the product of any two stabilizers $S_i$ and $S_j$ is itself a stabilizer.
Furthermore, any state can be defined by a set of \emph{stabilizer generators} $\mathcal{G}_\Psi$, which generates the group under multiplication, which we write as $\mathcal{S}_\Psi = \gen{\mathcal{G}_\Psi}$.


\emph{Stabilizer states} are further defined as a subset of the $n$-qubit states that can be efficiently described by a set of $n$ stabilizer generators 
\begin{equation}
	\mathcal{G}_{\Psi_S} = \{K_i: K_i\ket{\Psi_S} = \ket{\Psi_S}, K_i \in \mathcal{P}_n,\;i=1,\ldots,n\}
\end{equation}
where $\mathcal{P}_n$ is the group of $n$-fold tensor products of Pauli operators $\id$, $X$, $Y$ and $Z$ up to multiplicative phase factors $\pm 1$ and $\pm i$ and hence $\mathcal{S}_\Psi \subset \mathcal{P}_n$.
Specifically, stabilizer states are those produced by any stabilizer circuit which consists of only: i) preparation of qubits in computational basis states $\{\ket{0}, \ket{1}\}$; ii) quantum gates from the Clifford group\footnote
{
	Here we have used an alternative form of the Clifford group, replacing the conventional $CNOT$ with the $CZ$ gate, as they are equivalent up to $H$.
	This choice is in accordance with the graph state focus of this work and also provides a symmetric entangling operation that produces simpler update rules.
}
 $\mathcal{C} = \{H,\, S,\, CZ\}$; and iii) measurements in the computational basis.
The Gottesman-Knill theorem \cite{Gottesman1997} states that any such circuit can be simulated efficiently on a classical computer.


Stabilizer circuits include many that exhibit rich and canonically ``quantum'' phenomena such as superposition and entanglement, including the generation of large multipartite entangled states.
For such states, many correlations between measurement outcomes exist across the whole state, a fact that allows them to be used as quantum error correction codes \cite{Gottesman1997}.
Just as all correlations present in a state are represented in it's state vector, they are equally present in a state's stabilizers.


One can intuitively interpret the set of stabilizer generators $\mathcal{G}_\Psi$ as the minimal representation of the quantum correlations for $\ket{\Psi}$.
For example, consider the Bell state $\ket{\Phi^+}$
\begin{align}
	\ket{\Phi^+} = \frac{1}{\sqrt{2}}(\ket{00}+\ket{11}) = \frac{1}{\sqrt{2}}(\ket{++}+\ket{--}) = \frac{1}{\sqrt{2}}(\ket{+i-i} + \ket{-i+i}),
\end{align}
where $\ket{\pm} = \frac{1}{\sqrt{2}}(\ket{0} \pm \ket{1})$ and $\ket{\pm i} = \frac{1}{\sqrt{2}}(\ket{0} \pm i\ket{1})$.
By noting that $\ket{\Phi^+} = H_1\CZ_{1,2}H_1H_2\ket{00}$ it is easy to show that $\mathcal{S}_{\Phi^+} = \{X_1X_2, Z_1Z_2, -Y_1Y_2\} = \langle\mathcal{G}_{\Phi^+}\rangle = \langle  X_1X_2, Z_1Z_2 \rangle$, where $A_i$ represents the operator that enacts unitary $A$ on qubit $i$ and $\id$ everywhere else and similarly $A_{i,j}$ for two-qubit gates.
In this example it is clear that the stabilizers have provided the set all of correlations between single-qubit measurements on $\ket{\Phi^+}$, namely that the possible eigenvalues returned from measurements of both qubits in the $X$ and $Z$ basis are correlated ($\lambda_{X_1}\lambda_{X_2} = \lambda_{Z_1}\lambda_{Z_2} = 1$), whereas the possible eigenvalues found for $Y$ measurements are anti-correlated ($\lambda_{Y_1}\lambda_{Y_2} = -1$).


Just as the action of unitary operators evolve a state's quantum state vector in the Schr\"{o}dinger picture, a state's stabilizers are equivalently evolved within the Heisenberg picture \cite{Gottesman1997}.
The action of any Clifford gate unitary $U$ on a state $\ket{\Psi}$ therefore transforms $\mathcal{S}_\Psi$ as follows:
\begin{align}
	\ket{\Psi} \xrightarrow{U} \ket{{\Psi}^\prime} \quad \Leftrightarrow \quad \mathcal{S}_\Psi \xrightarrow{U} \mathcal{S}_{\Psi^\prime} = \{S_i^\prime = US_iU^\dagger : S_i \in \mathcal{S}_\Psi \} = \gen{K_i^\prime = UK_iU^\dagger : K_i \in \mathcal{G}_\Psi}.\label{eq:unitary_ev}
\end{align}


The effect of Pauli measurement operator $M \in \mathcal{P}_n$ on a state $\ket{\Psi}$ can also be represented by updating the stabilizer generators $\mathcal{G}_\Psi$.
For any $M$ there are two cases: either $M$ commutes with all of the state's stabilizers, or $M$ anti-commutes with one or more of them.
In the first case it is easy to show that either $M$ or $-M \in \mathcal{S}_\Psi$, and hence $\ket{\Psi}$ is an eigenstate of $M$ and so unaffected by the measurement.
However, in the latter case, the measurement $M$ will change the state.
In the case that the measurement $M$ returns an eigenvalue of $+1$, the stabilizers are updated as follows:
\begin{enumerate}
	\item 	Pick a $K_a \in \mathcal{G}_\Psi$ such that $\acomm{M}{K_a} = 0$.
			Replace $K_a$ with $M$.
	\item For all other $K_i \in \mathcal{G}_\Psi \!\setminus \!\{M\}$: 
	\begin{enumerate}
		\item If $\comm{M}{K_i} = 0$, leave $K_i$ unchanged
		\item If $\acomm{M}{K_i} = 0$, replace $K_i$ with $K_aK_i$.
	\end{enumerate}
\end{enumerate}
In the case that $M$ returns eigenvalue $-1$, the same process is applied, except $M \mapsto -M$ \cite{Gottesman1997}.


If the number of stabilizer generators on an $n$-qubit state $\ket{\Psi}$ is reduced from $n$, $\mathcal{G}_\Psi$ no longer defines a single state, but rather a subspace of states.
A set of logical basis states can be defined on this subspace together with logical operators that satisfy Pauli relations, thus creating an encoded \emph{logical} qubit.
While such constructions are commonly applied to design quantum error correcting codes, their application can be applied to other quantum information protocols, such as quantum teleportation.


In our case, we are specifically interested in the complete set of logical operators for a qubit input into some Clifford circuit $U$ (with some set of input ancillae qubits).
For example, consider an unknown state $\ket{\psi}_I$ of a single \emph{input} qubit $I$ with logical operators $\bar{X}_\psi = X_I$ and $\bar{Z}_\psi = Z_I$, which is then encoded via $\ket{\Psi} = U(\ket{\psi}_I\otimes\ket{0}^{\otimes n})$, such that 
\begin{align}
	\bar{X}_\Psi = UX_IU^\dagger, \quad \bar{Z}_\Psi = UZ_IU^\dagger
	\quad\textrm{and}\quad \mathcal{G}_{\Psi} = \{UZ_iU^\dagger\}_{i=1}^n
\end{align}
However, after encoding there are many other valid logical operators, as the product of a logical operator and stabilizer is also a valid logical operator.
Hence, the set of all logical operators for our encoded qubit is given by
\begin{align}
	\mathcal{L}_{\Psi} = i^k \times \{S\bar{L}_\Psi : S \in \mathcal{S}_{\Psi}, \bar{L}_\Psi \in \{\bar{X}_\Psi, \bar{Z}_\Psi, \bar{Y}_\Psi\}\}
\end{align}
where $i^k = \{1, -1, i, -i\}$ and $\bar{Y}_\Psi = i\bar{X}_\Psi\bar{Z}_\Psi$.
Formally, $\mathcal{L}_{\Psi}$ is the centralizer subgroup of operators in $\mathcal{P}_n$ which commute with the stabilizers of $\ket{\Psi}$.
Just as with stabilizers, the logical operators are similarly updated after a measurement $M$.
If $\comm{\bar{L}}{M}=0$ for $\bar{L} \in \mathcal{L}_{\Psi}$, then $\bar{L}$ is unchanged, otherwise the logical operators are transformed by $\bar{L} \mapsto \bar{L}^\prime = K_a\bar{L}$.

\section{Algorithm details} \label{apx:SPF}

Here we present a full description of the update rules applied in our algorithm to achieve stabilizer pathfinding.
First, we provide an algorithm that allows for stabilizers' and generators' triviality to be more efficiently tested.
Secondly, we provide update rules to track a state's stabilizers for any Clifford circuit.
Finally, we present a method for combining said stabilizers to produce valid teleportation measurement patterns.
The summary pseudocode for the above algorithms are also presented in algorithm \ref{alg:SPF}.

\subsection{Triviality tests} \label{apx:SPF:triv_testing}

\subsubsection{Testing stabilizers' triviality} \label{apx:SPF:triv_testing:stabs}


As part of the algorithm we shall describe, it will be necessary to remove some unknown trivial stabilizers, namely after applying a $\CZ$ or measurement.
In the general case of deciding whether some arbitrary stabilizer $S_c$ is trivial or not, given only $\mathcal{G}_\Psi$, the authors could not improve on a limited exhaustive search.
In this case, the space of all bipartitions is explored by finding $\mathcal{B}_c = \{S_b : b \in \mathcal{P}(c), \abs{b} \leq \floor*{c / 2}\}$, where $\mathcal{P}(c)$ is the power set of $c$ and $\{b : \abs{b} \leq \floor*{c / 2}\}$ being the set of all smaller halves of every possible bipartition of $c$.
For each element of $\mathcal{B}$, $S_bS_c$ is found and if $\qsupp{S_bS_c} \cap \qsupp{S_b} = \emptyset$, then $(b, c \setminus b)$ describes a trivial bipartition of $S_c$.
If no such $b$ is found, then $S_c$ must be non-trivial.
Clearly this method---which we refer to as \emph{single-shot triviality testing}---becomes inefficient for large $\abs{c}$.


However, a significantly faster triviality test can be performed within the context of our stabilizer pathfinding algorithm.
Consider the case where you have a large set of stabilizers $\mathcal{S}_\Psi^*$, some of which are trivial $\mathcal{S}_\Psi^t \subseteq \mathcal{S}_\Psi^T$, but which is otherwise guaranteed to contain all non-trivial stabilizers, such that $\mathcal{S}^*_{\Psi} = \mathcal{S}_\Psi^t \cup \mathcal{S}_\Psi^{\mathrm{NT}}$.
The task is then to extract $\mathcal{S}_\Psi^{\mathrm{NT}}$ by removal of $\mathcal{S}_\Psi^t$ without an exhaustive search.
To do so, initially consider testing a single $S_c \in \mathcal{S}_\Psi^*$ for triviality.
If $S_c$ is trivial, there must exist some \emph{minimal} bipartition $(\alpha, \beta)$ of $c$ such that either $S_\alpha$ and/or $S_\beta$ are non-trivial.
Since $\mathcal{S}_\Psi^{\mathrm{NT}} \subseteq \mathcal{S}_\Psi^*$, any such bipartitions can be identified by finding $\mathcal{B}^*_c = \{S_\beta : \beta \subset c, \abs{\beta} \leq \floor*{c / 2}, S_\beta \in \mathcal{S}_\Psi^*\}$ and then tested using the same process as single-shot triviality testing\footnote
{
	To further increase the efficiency of this test, $\beta$ are tested in ascending cardinality.
	Hence, if $\mathcal{B}$ does contain trivial $S_\beta$, they are never tested, as all $\alpha \subset \beta$ are tested first (although they would still correctly identify a trivial bipartition if tested).
}.
By repeating all $S_c \in \mathcal{S}_\Psi^*$ and removing any that fail, $\mathcal{S}_\Psi^{\mathrm{NT}}$ can be found with less than $\mathcal{O}(\abs{\mathcal{S}_\Psi^*}^2)$ tests (and far fewer in practise).
We shall refer to this type of triviality testing as \emph{batch} triviality testing.

\subsubsection{Testing generators' triviality} \label{apx:SPF:triv_testing:gens}
	
In rare cases, single-qubit measurements can cause generators themselves to become trivial.
As triviality of stabilizers is assessed under the assumption of generator non-triviality, these trivial generators must be detected and replaced, in a process known as \emph{generator detrivialisation}.
Specifically, we consider the case when there exists some generator $K_a$ and stabilizer $S_b$ for $a \not\in b$ such that $\qsupp{K_aS_b} \cap \qsupp{S_b} = \emptyset$ (recall that $\qsupp{A}$ is the set of qubits on which $A$ non-trivially acts).

For example, consider the 5-qubit stabilizer state\footnote
{
	Where $\ket{\Psi}$ can be produced by the Pauli $X$ measurement of a qubit within a 6-qubit ring graph state.
}
 $\ket{\Psi}$ that undergoes measurement $X_4$ as follows:
\begin{center}
\begin{minipage}{0.35\linewidth}
	\vspace{10pt}
	\begin{align}	
		\arraycolsep=0.7pt
		\begin{array}{rlllllllc}
			\bar{X}_\Psi = &     & Z_0 &     &     &     &     &    & \\
			\bar{Z}_\Psi = &     & X_0 &     &     & Z_3 & Z_4 &    & \\
		\mathcal{G}_\Psi = & \{  &     & X_1 & X_2 & Z_3 & Z_4 & ,  & \; (K_1) \\
						   &     &     & Z_1 & Z_2 &     &     & ,  & \; (K_2) \\
			   			   &     & Z_0 &     & Z_2 & X_3 &     & ,  & \; (K_3) \\ 
					       &     & Z_0 & Z_1 &     &     & X_4 & \} & \; (K_4) 
		\end{array} \notag
	\end{align}
\end{minipage}
\begin{minipage}{0.1\linewidth}
	$\xrightarrow{\textrm{Measure }X_4}$
\end{minipage}
\begin{minipage}{0.35\linewidth}
	\vspace{10pt}
	\begin{align}	
		\arraycolsep=0.7pt
		\begin{array}{rlllllllc}
			\bar{X}_\Psi = &     & Z_0 &     &     &     &     &    & \\
			\bar{Z}_\Psi = &     & X_0 & X_1 & X_2 &     &     &    & \\
		\mathcal{G}_\Psi = & \{  &     &     &     &     & X_4 & ,  & \; (K_1) \\ 
			 	       	   &     &     & Z_1 & Z_2 &     &     & ,  & \; (K_2) \\
			   			   &     & Z_0 &     & Z_2 & X_3 &     & ,  & \; (K_3) \\ 
					       &     & Z_0 & Z_1 &     &     &     & \} & \; (K_4) 
		\end{array} \notag
	\end{align}
\end{minipage}
\end{center}
Here we observe that after measurement $K_3 = Z_0 Z_2 X_3$, whereas $S_{\{2,3,4\}} = X_3$ and so $K_3$ is now trivial (or equivalently that $\qsupp{S_{\{2,3,4\}}} \cap \qsupp{S_{\{2,4\}}} \neq \emptyset$ before the measurement whereas $\qsupp{S_{\{2,3,4\}}} \cap \qsupp{S_{\{2,4\}}} = \emptyset$ after).
To ensure all generators are non-trivial they are updated such that $K^\prime_3 = K_2K_3K_4 = X_3$ and $K^\prime_i = K_i$ otherwise.

However, because each stabilizer's combination now represents a different set of generators, this update can change bipartitions' support overlaps and so updating the set of non-trivial stabilizers is more involved.
Firstly, stabilizers that do not contain the updated generator are unaffected, such that $S^\prime_c=S_c$ for $a \not\in c$ and hence their triviality is also unchanged.
However, during detrivialisation previously trivial stabilizers $S_c$ may become non-trivial $S_c^\prime$ if $a \in c$.
To find all newly non-trivial stabilizers, all stabilizer pairs $S_\alpha, S_\beta \in \mathcal{S}_\Psi^{\mathrm{NT}}$, $\alpha \cap \beta = \emptyset$ where
\begin{align}
	\qsupp{S_\alpha} \cap \qsupp{S_\beta} = \emptyset \quad \textrm{but} \quad \qsupp{S^\prime_\alpha} \cap \qsupp{S^\prime_\beta} \neq \emptyset.
\end{align}
are found and $S^\prime_{\alpha\cup\beta}$ added to $\mathcal{S}_{\Psi^\prime}^{\mathrm{NT}}$.
This process is then repeated on the newly non-trivial stabilizers found to ensure all previously trivial tripartitions, etc.\ are found.
Finally, any trivial stabilizers are then removed by applying a batch triviality test on all stabilizers $S^\prime_c \in \mathcal{S}_{\Psi^\prime}^{\mathrm{NT}}$ with $a \in c$.

We lastly note that trivial generators are an unavoidable by-product of the multiplication of generators performed after measurement and hence is never required after any unitary operation.

\subsection{SPF Part 1: Tracking all non-trivial stabilizers} \label{apx:SPF:algorithm_p1}


Stabilizer pathfinding must track the action of the three elements of any Clifford circuit, namely: 
\begin{enumerate}[label=\roman*)]
	\item preparation of qubits in computational basis states $\{\ket{0}, \ket{1}\}$;
	\item quantum gates from the Clifford group $\mathcal{C} = \{H,\, S,\, CZ\}$; and
	\item measurements in the computational basis,
\end{enumerate}
on the state's non-trivial combinations $\mathcal{S}_\Psi^{\mathrm{NT}}$.
We shall define the action of these operations as a series of set update rules using the convention $A_{\Psi} \mapsto A_{\Psi^\prime}$.

\subsubsection{Single-qubit gates} \label{apx:SPF:algorithm_p1:single_qubits}


Firstly, consider appending qubit $\ket{0}_{n+1}$ to $\ket{\Psi_S}$.
The stabilizer generators are simply updated by
\begin{equation}
	\mathcal{G}_\Psi \mapsto \mathcal{G}_{\Psi^\prime} = \mathcal{G}_\Psi \cup \{K_{n+1}\},
\end{equation}
where $K_{n+1} = Z_{n+1}$.
Since $\qsupp{K_{n+1}} \cap \qsupp{K_i} = \emptyset$ for all $K_i \in \mathcal{G}_\Psi$, there is only a single new non-trivial combination, namely $\{n+1\}$, and hence the non-trivial stabilizers are similarly updated by
\begin{equation}
	\mathcal{S}_\Psi^{\mathrm{NT}} \mapsto \mathcal{S}_{\Psi^\prime}^{\mathrm{NT}} = \mathcal{S}_\Psi^{\mathrm{NT}} \cup \{S_{\{n+1\}}\}
\end{equation}


Next, consider the action of quantum gates from $\mathcal{C}$ on an $n$-qubit stabilizer state.
As defined in section \ref{apx:stab_theory}, when acted on by $U \in \mathcal{C}$ the stabilizer generators are simply updated as in equation \eqref{eq:unitary_ev}, or
\begin{equation}
	\mathcal{G}_\Psi \mapsto \mathcal{G}_{\Psi^\prime} = \{K_i^\prime\}_{i=1}^{n} = \{UK_iU^\dagger\}_{i=1}^{n}, \label{eq:gens_unitary_update}
\end{equation}
where $K_i \in \mathcal{G}_\Psi \;\forall\; i$.
In the case of a single-qubit Clifford gate $U \in \{H,\, S\}$, it is simple to show that all non-trivial stabilizers can be similarly updated via
\begin{equation}
	\mathcal{S}_\Psi^{\mathrm{NT}} \mapsto \mathcal{S}_{\Psi^\prime}^{\mathrm{NT}} = \{S_c^\prime\} = \{US_cU^\dagger\}.
\end{equation}
Importantly, for the above statement to hold, it must also true that all non-trivial stabilizers remain non-trivial after applying $U$ and similarly for those that are trivial.
This requirement is shown to hold in the following Remark.

\begin{remark} \label{apx:lems:single_q_unitaries}
	If before the action of a single-qubit Clifford gate $U$ a stabilizer is trivial (non-trivial), $S_c \in \mathcal{S}_\Psi^{\mathrm{T}}$ ($\mathcal{S}_\Psi^{\mathrm{NT}}$), then after $U$ it remains trivial (non-trivial), $S_c^\prime \in \mathcal{S}_{\Psi^\prime}^{\mathrm{T}}$ ($\mathcal{S}_{\Psi^\prime}^{\mathrm{NT}}$).
\end{remark}
\begin{proof}
	Firstly, we consider the case of a trivial stabilizer $S_c$.
	If $S_c \in \mathcal{S}_\Psi^{\mathrm{T}}$, there exists a trivial bipartition $(\alpha, \beta)$ of $c$ such that 
	\begin{equation}
		S_c = S_\alpha S_\beta = \prod_{i \in \alpha}K_i \prod_{j \in \beta}K_j \quad \quad\textrm{and}\quad \quad \qsupp{S_\alpha} \cap \qsupp{S_\beta} = \emptyset.\label{eq:trivial_combo}
	\end{equation}
	Without loss of generality, assume $\qsupp{U} \subseteq \qsupp{S_\alpha}$ (since $\abs{\qsupp{U}} = 1$).
	From equation \eqref{eq:gens_unitary_update}, then $S_\alpha^\prime = US_\alpha U^\dagger$ and $S_\beta^\prime = US_\beta U^\dagger = S_\beta$.
	Finally, since $\qsupp{S_\alpha^\prime} = \qsupp{US_\alpha U^\dagger} = \qsupp{S_\alpha}$, then from equation \eqref{eq:trivial_combo}, $\qsupp{S_\alpha^\prime} \cap \qsupp{S_\beta^\prime} = \emptyset$ and hence $(\alpha, \beta)$ is also a trivial bipartition of $c$ after $U$, showing that $S_c^\prime \in \mathcal{S}_{\Psi^\prime}^{\mathrm{T}}$.
	Secondly, in the non-trivial case, the above proof can be easily inverted to show that if after $U$, $S_c^\prime \in \mathcal{S}_{\Psi^\prime}^{\mathrm{T}}$ then $S_c$ must also admit a trivial bipartition, and hence $S_c \notin \mathcal{S}_{\Psi}^{\mathrm{NT}}$.
	It follows that $S_c \in \mathcal{S}_{\Psi}^{\mathrm{NT}} \Rightarrow S_c^\prime \in \mathcal{S}_{\Psi^\prime}^{\mathrm{NT}}$.
	
\end{proof}

\subsubsection{Two-qubit gates} \label{apx:SPF:algorithm_p1:cz_gates}


Now consider the $\CZ_{u,v}$ gate applied to qubits $u$ and $v$.
For $\CZ_{u,v}$ the stabilizer generators are similarly updated using equation \eqref{eq:gens_unitary_update}, however it is also possible that new non-trivial stabilizers are produced and/or previously non-trivial stabilizers are made trivial.
In general, this will cause the number of non-trivial stabilizers to change, for example, $\mathcal{S}^\mathrm{NT}_\Psi = \{X_1, X_2\}$ for the empty two-qubit graph state, whereas $\mathcal{S}^\mathrm{NT}_{\Psi^\prime} = \{X_1Z_2, Z_1X_2, Y_1Y_2\}$ after $\CZ_{1,2}$ is applied.
Since the effect of $\CZ_{u,v}$ on any $K_i$ may either increase or decrease $\abs{\qsupp{K_i}}$ there are two cases that must be considered for stabilizer pathfinding: either stabilizers that go from i) non-trivial to trivial or ii) trivial to non-trivial.
A method for efficiently updating $\mathcal{S}_\Psi^{\mathrm{NT}}$ is now presented below.


First we address i), the case of $\CZ_{u,v}$ causing previously non-trivial stabilizers to become trivial, where for some stabilizer initially $S_c \in \mathcal{S}_\Psi^{\mathrm{NT}}$ but $S_c^\prime \in \mathcal{S}_{\Psi^\prime}^{\mathrm{T}}$ afterwards.
Given $\CZ_{u,v}$ can change the qubit support of any stabilizer $S_i$ by at most a single qubit $u$ or $v$ then it must be true for $S_c \in \mathcal{S}_\Psi^{\mathrm{NT}}$ that there exists some bipartition $(\alpha, \beta)$ of $c$ such that
\begin{align}
	&\qsupp{S_\alpha} \cap \qsupp{S_\beta} \subseteq \{u, v\} \\ 
	\textrm{but} \quad &\qsupp{S_\alpha^\prime} \cap \mathcal{Q}(S_\beta^\prime) = \emptyset
\end{align}
where $S_\alpha, S_\beta \in \mathcal{S}_\Psi^{\mathrm{NT}}$.
Since $\mathcal{S}_\Psi^{\mathrm{NT}}$ is known, finding $(\alpha, \beta)$ requires finding $S_\alpha, S_\beta \in \mathcal{S}_\Psi^{\mathrm{NT}}$ where the decrease in support of $S_\alpha$ and $S_\beta$ is exactly equal to the previous support overlap between them.


Next we consider ii), the case of $\CZ_{u,v}$ causing previously trivial stabilizers to become non-trivial, where for some stabilizer initially $S_c \in \mathcal{S}_\Psi^{\mathrm{T}}$ but after $S_c^\prime \in \mathcal{S}_{\Psi^\prime}^{\mathrm{NT}}$.
For this case, an initial search must be performed to find some bipartition $(\alpha,\beta)$ for which
\begin{align}
&\qsupp{S_\alpha} \cap \mathcal{Q}(S_\beta) = \emptyset \label{eq:t-nt_no_init_supp}\\
\textrm{but}\quad &\qsupp{S_\alpha^\prime} \cap \mathcal{Q}(S_\beta^\prime) \subseteq \{u, v\} \label{eq:t-nt_final_supp}
\end{align}
where $S_\alpha, S_\beta \in \mathcal{S}_\Psi^{\mathrm{NT}}$.
Similarly to i), this can only be achieved if the increase in support is equal to the new support overlap.
While equation \eqref{eq:t-nt_final_supp} is a necessary condition for non-triviality, it is not sufficient as it must also hold across all possible bipartitions.
In some cases it may be simultaneously possible to find two bipartitions of $c$, with one $(\alpha, \beta)$ satisfying equations \eqref{eq:t-nt_no_init_supp} and \eqref{eq:t-nt_final_supp} and another $(\gamma, \delta)$ that does not.
Fortunately, these cases can be easily detected and there are three possible variants:
\begin{enumerate}[label=\alph*)]
	\item $\qsupp{{S^\prime_\gamma}} \cap \{u,v\} = \qsupp{{S^\prime_\delta}} \cap \{u,v\} = \emptyset$ (neither has support on $u, v$),
	\item $\qsupp{{S^\prime_\gamma}} \cap \{u,v\} = \{u\}$, $\qsupp{{S^\prime_\delta}} \cap \{u,v\} = \{v\}$ (both are supported on $u, v$, but without overlap),
	\item $\qsupp{{S^\prime_\gamma}} \cap \{u,v\} = \emptyset$, $\qsupp{{S^\prime_\delta}} \cap \{u, v\} \neq \emptyset$ (only one has support on $u, v$).
\end{enumerate}

For a), $S_c = S_\gamma S_\delta \Rightarrow u, v \notin \qsupp{S_c}$ and Remark \ref{apx:lems:no_u_v_supp} can be applied to show that if such a bipartition does exist then equations \eqref{eq:t-nt_no_init_supp} and \eqref{eq:t-nt_final_supp} cannot be simultaneously satisfied and hence a) never occurs.

\begin{remark}\label{apx:lems:no_u_v_supp}
	If before the action of $\CZ_{u,v}$, a stabilizer $S_c$ where $u, v \notin \qsupp{S_c}$ is trivial (non-trivial), then after $\CZ_{u, v}$ it remains trivial (non-trivial).
\end{remark}

\begin{proof}
	If $u, v \notin \qsupp{S_c}$ then it must be the case that ${S_\alpha}^{[u]} = {S_\beta}^{[u]}$ and ${S_\alpha}^{[v]} = {S_\beta}^{[v]}$ for all possible bipartitions $(\alpha, \beta)$, where recall $A^{[i]}$ is the Pauli operator of $A$ acting on qubit $i$.
	Hence after $\CZ_{u,v}$ then ${S_\alpha}^{[u]} = {S_\beta}^{[u]}$ and ${S_\alpha}^{[v]} = {S_\beta}^{[v]}$, and so $u, v \notin \qsupp{S^\prime_c}$.
	For a stabilizer to become trivial from non-trivial (or vice versa), then it must be true that some bipartition must change from sharing support to not sharing support (or vice versa).
	However, it follows immediately from the previous comments that for any bipartition $(\alpha, \beta)$ then
	\begin{align}
		\qsupp{S_\alpha} \cap \qsupp{S_\beta} = \emptyset \;\Leftrightarrow\;\qsupp{S^\prime_\alpha} \cap \qsupp{S^\prime_\beta} = \emptyset
	\end{align}
	and hence trivial and non-trivial stabilizers respectively remain so.
\end{proof}

For b) and c), because all $S_\gamma, S_\delta \in \mathcal{S}_\Psi^{\mathrm{NT}}$ that gain support are considered by the initial search for all $(\alpha, \beta)$ satisfying equations \eqref{eq:t-nt_no_init_supp} and \eqref{eq:t-nt_final_supp}, only $S_\gamma, S_\delta$ with no support gain are relevant here.
Additionally, since only one half of a trivial bipartition need be found to prove triviality, then we can further limit our search to stabilizers with some support on $u, v$, reducing the set of potentially trivial partitions of any stabilizer found by equations \eqref{eq:t-nt_no_init_supp} and \eqref{eq:t-nt_final_supp} to 
\begin{align}
	\mathcal{S}^*_{\Psi^\prime} = \{S^\prime_i : \qsupp{S^\prime_i} \cap \{u, v\} \neq \emptyset, S_i \in \mathcal{S}_{\Psi}^{\mathrm{NT}}\}.
\end{align} 
As this set is can be easily found from $\mathcal{S}_{\Psi}^{\mathrm{NT}}$, a batch triviality test can be applied to $\mathcal{S}^*_{\Psi^\prime}$ with the reduced partition batch $\mathcal{B}^*_c = \{S_b : b \subset c, \abs{b} \leq \floor*{c / 2}, S_b \in \mathcal{S}_{\Psi^\prime}^*\}$, allowing any trivial bipartitions to be found with minimal overhead cost.


To summarise, after $\CZ_{u,v}$, $\mathcal{S}_{\Psi}^{\mathrm{NT}}$ is updated by applying the following steps:
\begin{enumerate}
	\item Update all $S^\prime_c \in \mathcal{S}_{\Psi^\prime}^{\mathrm{NT}}$ with non-trivial support on $u$ and/or $v$, via $S^\prime_c \mapsto \CZ_{u,v}S_c\CZ_{u,v}$.
	\item Remove from $\mathcal{S}_{\Psi^\prime}^{\mathrm{NT}}$ any $S^\prime_c$ that admits a bipartition no longer containing support overlap.
	\item Add to $\mathcal{S}_{\Psi^\prime}^{\mathrm{NT}}$ any new $S^\prime_c$ that can be produced by stabilizer pairs that now share support.
	\item Apply a batch triviality test to $\mathcal{S}_{\Psi^\prime}^{\mathrm{NT}}$ with reduced partition batch $\mathcal{B}^*_c$ to remove any trivial stabilizers.
\end{enumerate}

\subsubsection{Qubit measurement} \label{apx:SPF:algorithm_p1:measurement}


We shall consider the general case of performing an arbitrary singe-qubit Pauli measurement $M \in \{X, Y, Z\}$ (returning a $+1$ eigenvalue).
In the standard approach to updating $\mathcal{G}_\Psi$, as described in Appendix \ref{apx:stab_theory} and Ref.\ \cite{Gottesman1997}, generators for which $\acomm{K_i}{M} = 0$ are updated as $K_i^\prime = K_aK_i$ for some chosen $\acomm{K_a}{M} = 0$ (with $K_a^\prime = M$).
However, after this update is applied, $K_i^\prime = K_aK_i$ may now be a trivial generator with respect to $M$, such that $\qsupp{MK_aK_i} \cap \qsupp{M} = \emptyset$.
In these cases, rather than applying the generator detrivialisation described section \ref{apx:SPF:triv_testing:gens}, we can apply a modified update to the generators $K_i^\prime = MK_aK_i$.
Similar remarks can apply for some cases where $\comm{K_i}{M} = 0$, in which case the update rule $K_i^\prime = MK_i$ is applied.
To summarise, after measurement $M$, the state's stabilizer generators are updated $\mathcal{G}_\Psi \mapsto \mathcal{G}_{\Psi^\prime} = \{K_i^\prime\}$ using the five following rules:
\[
\begin{array}{ll}
	K_a^\prime = M \quad &\textrm{for some} \; \acomm{K_a}{M} = 0 \\
	K_i^\prime = K_aK_i \quad &\textrm{if} \; \acomm{K_i}{M} = 0,\; K_i \neq K_a,\; \qsupp{MK_aK_i} \cap \qsupp{M} \neq \emptyset \\
	K_i^\prime = MK_aK_i \quad &\textrm{if} \; \acomm{K_i}{M} = 0,\; K_i \neq K_a,\; \qsupp{MK_aK_i} \cap \qsupp{M} = \emptyset \\
	K_i^\prime = K_i \quad &\textrm{if} \; \comm{K_i}{M} = 0,\; \qsupp{MK_i} \cap \qsupp{M} \neq \emptyset \\
	K_i^\prime = MK_i \quad &\textrm{if} \; \comm{K_i}{M} = 0,\; \qsupp{MK_i} \cap \qsupp{M} = \emptyset
\end{array}
\]
We further define two key sets of updated generators $A$ and $B$, such that $A = \{i: K_i^\prime = K_aK_i\} \,\cup\, \{i: K_i^\prime = MK_aK_i\}$ and $B = \{i: K_i^\prime = MK_i\} \,\cup\, \{i: K_i^\prime = MK_aK_i\}$.
From $A$ and $B$ we can derive the general update rule for arbitrary post-measurement stabilizers
\begin{align}
	S_c^\prime &= M^{\abs{c\,\cap\, \{a\}}}\left(\prod_{i \in c\,\cap\, (B\setminus A)}\!\!\!\!\!\!\!\!\!MK_i\right)\left(\prod_{j \in c\,\cap\, (A \cap B)}\!\!\!\!\!\!\!\!\!MK_aK_j\right)\left(\prod_{k \in c\,\cap\, (A \setminus B)}\!\!\!\!\!\!\!\!\!K_aK_k\right)\left(\prod_{l \in c \setminus (A \cup B \cup \{a\})}\!\!\!\!\!\!\!\!\!K_l\right) \label{eq:new_stab_decomp}\\
 	&= M^{\abs{c\,\cap\, \{a\}}}M^{\abs{c\,\cap\, B}}K_a^{\abs{c\,\cap\, A}} \prod_{i \in c \setminus \{a\}}\!\!\!\!K_i \\
 	&= M^{\abs{c\,\cap\, (B \cup \{a\})}}K_a^{\abs{c\,\cap\, A}} S_{c\setminus \{a\}}
\end{align}
where $\abs{A}$ denotes the cardinality of the set $A$, $A \setminus B$ the set difference of $A$ and $B$, and we have used the fact that $\comm{M}{MK_i} = 0\;\forall\; i \in B\setminus A$ and $\comm{K_i}{K_j} = 0 \;\forall\; i, j$.
From this, $S_c \in \mathcal{S}_\Psi^{\mathrm{NT}}$ can be easily updated.
However, given that a single measurement $M$ may change the support of many generators, updating $\mathcal{S}_\Psi^{\mathrm{NT}}$ finding newly trivial and non-trivial stabilizers is more involved.

Firstly, in the following description of measurement update rules, we will require the following Lemma: 
\begin{lemma}\label{apx:lems:new_non-triv_post-mnt}
	After single-qubit measurement $M$ is made on state $\ket{\Psi}$, all new non-trivial stabilizers, are contained within the set $\{S_c: a \notin c, S_{c \cup \{a\}} \in \mathcal{S}_{\Psi}^{\mathrm{NT}}\}$ and where $K_a$ is the generator removed from $\mathcal{G}_\Psi$ and replaced with $M$.
\end{lemma}

\begin{proof}
	Firstly, since $\abs{\qsupp{M}} = 1$ and $K_a^\prime = M$, then $S_c^\prime \in \mathcal{S}_{\Psi^\prime}^{\mathrm{T}}$ for all $c \ni a$.
	Hence only combinations that do not contain $a$ need be considered.
	
	Now consider the previously trivial combination $c \ni a$ with bipartition $(\alpha, \beta)$ such that $\qsupp{S_\alpha} \cap \qsupp{S_\beta} = \emptyset$ for $S_c = S_\alpha S_\beta \in \mathcal{S}_\Psi^{\mathrm{T}}$ before measurement.
	As in equation \eqref{eq:new_stab_decomp}, we can write the updated stabilizer as
	\begin{align}
		S_c^\prime &= S_\alpha^\prime \cdot S_\beta^\prime = \left(M^{\abs{\alpha \cap B}} K_a^{\abs{\alpha \cap A}} S_\alpha\right) \cdot \left(M^{\abs{\beta \cap B}} K_a^{\abs{\beta \cap A}} S_\beta\right)
	\end{align}
	
	We first consider the cases in which $\abs{\alpha \cap B}$ and $\abs{\beta \cap B}$ are even, for which there are three further sub-cases: 
	\begin{enumerate}[label=\roman*)]
		\item $\abs{\alpha \cap A}$ and $\abs{\beta \cap A}$ even $\Rightarrow$ $S_\alpha^\prime = S_\alpha$ and $S_\beta^\prime = S_\beta$;
		\item $\abs{\alpha \cap A}$ and $\abs{\beta \cap A}$ odd $\Rightarrow$ $S_\alpha^\prime = K_aS_\alpha$ and $S_\beta^\prime = K_aS_\beta$;
		\item $\abs{\alpha \cap A}$ odd and $\abs{\beta \cap A}$ even $\Rightarrow$ $S_\alpha^\prime = K_aS_\alpha$ and $S_\beta^\prime = S_\beta$ (and vice versa).
	\end{enumerate}
	For i), $S_\alpha^\prime = S_\alpha$ and $S_\beta^\prime = S_\beta \,\Rightarrow\, \qsupp{S_\alpha^\prime} \cap \qsupp{S_\beta^\prime} = \emptyset$, and $S_c^\prime \in \mathcal{S}_{\Psi^\prime}^{\mathrm{T}}$ remains trivial.
	For ii), if $S_\alpha \mapsto S_\alpha^\prime = K_aS_\alpha$, then $\acomm{S_\alpha}{M} = 0$ and hence $\qsupp{M} \subset \qsupp{S_\alpha}$.
	Since the same applies for $S_\beta^\prime$, then $\qsupp{S_\alpha} \cap \qsupp{S_\beta} \neq \emptyset$ and therefore $(\alpha, \beta)$ is not a trivial bipartition of $S_c$, which is a contradiction and so ii) does not occur.
	Finally for iii), $S_c^\prime = K_aS_\alpha \cdot S_\beta = S_{\alpha \,\cup\, \{a\}} \cdot S_\beta = S_{c \,\cup\, \{a\}}$, and therefore $(\alpha, \beta)$ is only a trivial bipartition for $S_c^\prime$ if $(\alpha \cup \{a\}, \beta)$ is for $S_{c \,\cup\, \{a\}}$.
	
	In the cases in which $\abs{\alpha \cap B}$ and/or $\abs{\beta \cap B}$ are odd, we observe that the effect of applying $M$ is to remove support on $\qsupp{M}$.
	The previous three cases then also apply except with $\qsupp{S_\alpha^\prime} \mapsto \qsupp{MS_\alpha^\prime} \subset \qsupp{S_\alpha^\prime}$ (and similarly for $\beta$) which can only decrease the number of cases where $\qsupp{S_\alpha^\prime} \cap \qsupp{S_\beta^\prime} = \emptyset$.
	
	It therefore follows that the only instances of trivial $S_c$ and non-trivial $S_c^\prime$ that occur are those for which $S_{c \cup \{a\}}$ is non-trivial.
	Or equivalently, if $S_c \in \mathcal{S}_{\Psi}^{\mathrm{T}}$ then $S_c \in \mathcal{S}_{\Psi^\prime}^{\mathrm{NT}}$ iff $S_{c \cup \{a\}} \in \mathcal{S}_{\Psi}^{\mathrm{NT}}$.
\end{proof}

We now proceed with the description of measurement update rules.
After the single-qubit measurement $M$ all stabilizer combinations containing $a$ become trivial, since $\abs{\qsupp{K^\prime_a}} = 1$, and so all $S^\prime_c$ with $a \in c$ are removed.
Next, using Lemma \ref{apx:lems:new_non-triv_post-mnt} we show that for any single-qubit Pauli measurement $M$ then all new non-trivial stabilizers are contained within the set $\Gamma = \{S_c: a \not\in c, S_{c \cup \{a\}} \in S_{\Psi}^{\mathrm{NT}}\}$.
Since all previously non-trivial $S_{c \cup \{a\}}\in S_{\Psi}^{\mathrm{NT}}$ stabilizers are known, $\Gamma$ is easily found.
However, not all stabilizers in $\Gamma$, nor the remaining updated stabilizers $\Delta = \{S_c^\prime : S_c \in S_\Psi^{\mathrm{NT}}\} \setminus \{S_c^\prime : a \in c\}$, are necessarily still non-trivial.
But since $\mathcal{S}_{\Psi^\prime}^{\mathrm{NT}} \subseteq \Gamma \cup \Delta$, all trivial stabilizers can identified and removed using the batch triviality test described in section \ref{apx:SPF:triv_testing:stabs}.

Finally, we note that in some cases measurement $M$ can cause the updated generators to be trivial in a manner not captured above.
For example, in the simplest case, this can occur when performing $M$ also performs the indirect single-qubit measurement $M^\prime$.
However, this can be easily identified as a non-trivial replacement for the trivial generator will always be contained within $S_\Psi^{\textrm{NT}}$.
Once any trivial generators are identified, the process described in section \ref{apx:SPF:triv_testing:gens} can then be applied to return $\mathcal{G}_\Psi$ to it's proper form.

To summarise, $\mathcal{S}_{\Psi^\prime}^{\mathrm{NT}}$ is updated by applying the following steps:
\begin{enumerate}
	\item Remove from $\mathcal{S}_{\Psi\prime}^{\textrm{NT}}$ any $S_c$ with $a \in c$, but keeping them in memory.
	\item Using the discarded stabilizers, find and add the set of potential new non-trivial stabilizers $\{S_c : S_{c \cup \{a\}} \in \mathcal{S}_\Psi^{\mathrm{NT}}, S_c \notin \mathcal{S}_\Psi^{\mathrm{NT}}\}$.
	\item Apply a batch triviality test to $\mathcal{S}_{\Psi^\prime}^{\mathrm{NT}}$ to remove any trivial stabilizers.
	\item Test generators for triviality and update if required.
\end{enumerate}

\subsection{SPF Part 2: Finding loss-tolerant measurement patterns} \label{apx:SPF:algorithm_p2}


We now consider the task of using the set of non-trivial stabilizers to identify a set of measurement patterns that teleport from $I$ to $O$, as outlined in section \ref{sec:SPF:finding_mnt_pats}.
In this case, each non-trivial stabilizer $S_a \in \mathcal{S}_\Psi^{\mathrm{NT}}$ represents three possible anti-commuting logical operators, namely $S_c\bar{X}$, $S_c\bar{Y}$ and $S_c\bar{Z}$.
For each valid teleportation measurement pattern, this logical operator must then be combined with another represented by a second stabilizer $S_d$ to satisfy equation \eqref{eq:motivation:log_op_pairs}.
For states of significant size and complexity, this presents a large number of possible measurement patterns, which are prohibitively expensive to calculate and validate.
However, due to requirement for achieving maximal loss tolerance, our algorithm is specifically concerned with measurement patterns that have minimal qubit weight, and only needs to consider a restricted subset of all patterns.
We shall now present an algorithm that finds all measurement patterns below a certain qubit weight given $\mathcal{S}_\Psi^{\mathrm{NT}}$, and hence the set of maximally loss-tolerant measurement patterns\footnote
{
	Strictly speaking, in certain cases logical operators can also be produced by combining multiple non-trivial stabilizers, for example if $S_c\bar{X}$ is a logical operator with support on $O$, but some other stabilizer exists where
	\begin{align}
		\qsupp{S_\alpha} \cap \qsupp{\bar{X}} \neq \emptyset \quad\textrm{but}\quad \qsupp{S_\alpha} \cap \qsupp{S_c} = \emptyset.
	\end{align}
	In this case, while $S_cS_\alpha \bar{X}$ does represent a valid logical operator, it is less loss-tolerant that $S_c \bar{X}$ since in almost all reasonable cases $\abs{\qsupp{S_cS_\alpha \bar{X}}} > \abs{\qsupp{S_c\bar{X}}}$.
	However, for the states in where these combinations do improve loss tolerance, they can be straightforwardly included.
	Alternatively, if a virtual qubit $I$ is being utilised by entangling it with some set qubits $\mathcal{I}$, it can be assumed that only those qubits within $\mathcal{I} \cap \qsupp{S_c}$ are initially entangled with $I$ so as to minimise any unnecessary measurements by removing any of the above cases.
	For these reasons we can safely omit these logical operators from our consideration.
}.


Following equation \eqref{eq:motivation:log_op_pairs}, the set of all measurement patterns $\mathcal{M}$ that achieve the desired teleportation is
\begin{align}
	\mathcal{M} = \{M_{\bar{L}_i, \bar{L}_j} : \acomm{\bar{L}_i^{[O]}}{\bar{L}_j^{[O]}} = 0,  \comm{\bar{L}_i^{[a]}}{\bar{L}_j^{[a]}} = 0, \;\forall\; a \neq O\}
\end{align}
Given that $\abs{M_{\bar{L}_i, \bar{L}_j}} \geq \max(\abs{\qsupp{\bar{L}_i}}, \abs{\qsupp{\bar{L}_j}}) - 1$, the set $\mathcal{M}_w$ of measurement patterns with weight $w$ is a subset of all $M_{\bar{L}_i, \bar{L}_j}$ produced by logical operators with at most weight $w + 1$, such that
\begin{align}
	\mathcal{M}_w \subset \{ M_{\bar{L}_i, \bar{L}_j} : \abs{\qsupp{\bar{L}_i}}, \abs{\qsupp{\bar{L}_j}} \leq w + 1\}.
\end{align}
Note that it is necessarily true that $\mathcal{M}_v \subseteq \mathcal{M}_w$ for $v < w$, and also that any measurement pattern that does not contain (i.e.\ is loss-tolerant to) a certain set of qubits, is also loss-tolerant to any subset of those qubits.
Hence, many loss-tolerant configurations on an $n$-qubit state can be found from combining logical operators with weight $w \ll n$.
In practise, $w$ can be increased until no additional loss tolerance is found, thereby providing all possible loss-tolerant configurations without the need for an exhaustive search\footnote
{
	This is specifically true for heralded loss.
	For unheralded loss, increasing $w$ may be beneficial in order to identify a greater variety of measurement patterns with the same degree of loss tolerance, thereby increasing the number of measurement patterns available after a given qubit measurement should it fail.
}.

\afterpage{
\begin{figure}[H]
	\begin{algorithm}[H]
		\SetAlgoLined
		\DontPrintSemicolon
		\SetKwProg{obj}{Object}{}{}
		\SetKwProg{meth}{Method}{}{}
		\SetKwProg{meth}{Method}{}{}
		\SetKw{KwLoopIn}{in}
		\obj{$\textsc{StabilizerState}$}
		{
			\vspace{5pt}
			\meth{$\textsc{InitialiseState}$}
			{
				$Q = \{I\}$ \tcp*[r]{Creates qubit register $Q$}
				$\bar{X} \leftarrow X_I, \; \bar{Z} \leftarrow Z_I$ \tcp*[r]{Creates logical operators}
				$\mathcal{G} \leftarrow \emptyset, \; \mathcal{S} \leftarrow \emptyset$ \tcp*[r]{Creates generators $\mathcal{G}$ and stabilizer combos $\mathcal{S}$}
			}
			\vspace{5pt}
			\meth{$\textsc{AddQubit}(q)$}
			{
				$Q \leftarrow Q \cup \{q\}$ \tcp*[r]{Appends qubit to register}
				$\bar{X} \leftarrow \bar{X} \otimes \id_q, \; \bar{Z} \leftarrow \bar{Z} \otimes \id_q$ \tcp*[r]{Extends logical ops.}
				$\mathcal{G} \leftarrow \{K_i \otimes \id_q : K_i \in \mathcal{G}\}, \;\; \mathcal{S} \leftarrow \{S_c \otimes \id_q : S_c \in \mathcal{S}\}$ \tcp*[r]{Extends gens.\ and combo stabs.}
			}
			\vspace{5pt}
			\meth{$\textsc{ApplyQubitUnitary}(q, U)$}
			{
				$\bar{X} \leftarrow U\bar{X}U^\dag, \; \bar{Z} \leftarrow U\bar{Z}U^\dag$ \tcp*[r]{Updates logical operators}
				$\mathcal{G} \leftarrow \{UK_iU^\dag : K_i \in \mathcal{G}\}, \;\; \mathcal{S} \leftarrow \{US_cU^\dag : S_c \in \mathcal{S}\}$ \tcp*[r]{Updates gens.\ and combo stabs.}
			}
			\vspace{5pt}
			\meth{$\textsc{ApplyCZ}(u, v)$}
			{	
				$\bar{X} \leftarrow \CZ_{uv}\bar{X}\CZ_{uv}^\dag, \; \bar{Z} \leftarrow \CZ_{uv}\bar{Z}\CZ_{uv}^\dag$ \tcp*[r]{Updates logical operators}
				$\mathcal{G}\; \leftarrow \{\CZ_{uv} K_i \CZ_{uv}^\dag : K_i \in \mathcal{G}\}$ \tcp*[r]{Updates gens.}
				$\mathcal{S}^\prime \leftarrow \{\CZ_{uv} S_c \CZ_{uv}^\dag : S_c \in \mathcal{S},\; S_c^{[u]} \otimes S_c^{[v]} \neq \id_u \otimes \id_v\}$ \tcp*[r]{Finds post-CZ combo stabs.}
				\For{$S_c^\prime$ \KwLoopIn $\mathcal{S}^\prime$}
				{
					\If{$\textsc{IsTrivial}(S_c^\prime)$}
					{
						$\mathcal{S}^\prime \leftarrow \mathcal{S}^\prime \setminus \{S_c^\prime\}$ \tcp*[r]{Removes trivial combo stabs.}
					}
				}
				\For{$S_c^\prime, S_d^\prime$ \KwLoopIn $\mathcal{S}^\prime$}
				{
					\If{$\qsupp{S_c} \cap \qsupp{S_d} = \emptyset \bigwedge \qsupp{S_c^\prime} \cap \qsupp{S_d^\prime} \subseteq \{u, v\}$}
					{
						$\mathcal{S}^\prime \leftarrow \mathcal{S}^\prime \cup \{S_c^\prime \cdot S_d^\prime\}$ \tcp*[r]{Adds newly non-trivial combo stabs.}
					}
				}
				$\mathcal{S}^\ast \leftarrow \{S_c^\prime : S_c \in \mathcal{S}, \qsupp{S_c^\prime} \cap \{u, v\} \neq \emptyset\}$ \;
				$\mathcal{S} \leftarrow \mathcal{S}^\prime \setminus \textsc{FindTrivStabs}(\mathcal{S}^\ast)$ \tcp*[r]{Removes triv.\ stabs.\ found in batch triv.\ test}
			}
			\vspace{5pt}
			\meth{$\textsc{MeasureQubit}(q, M)$}
			{
				$K_a \leftarrow \textsc{FirstAntiCommGen}(\mathcal{G}, M)$ \tcp*[r]{Picks anti-comm.\ gen.\ to remove}
				$\mathcal{G}^\prime \leftarrow \mathcal{G} \setminus \{K_a\}$ \tcp*[r]{Updates gens.}
				\For{$K_i^\prime$ \KwLoopIn $\mathcal{G}^\prime \setminus \{M\}$}
				{
					\If{$\acomm{K_i^\prime}{M} = 0$}
					{
						$b \leftarrow \textsc{Bool}(\qsupp{MK_aK_i} \cap \qsupp{M} = \emptyset)$ \tcp*[r]{$\textsc{Bool(True) = 1}, \textsc{Bool(False) = 0}$}
						$K_i^\prime \leftarrow M^bK_aK_i$
					}
					\ElseIf{$\comm{K_i^\prime}{M} = 0$}
					{
						$b \leftarrow \textsc{Bool}(\qsupp{MK_i} \cap \qsupp{M} = \emptyset)$ \;
						$K_i^\prime \leftarrow M^bK_i$
					}
				}
				$\mathcal{S}^\prime \leftarrow \mathcal{S} \setminus \{S_c : S_c \in \mathcal{S}, a \in c\}$ \tcp*[r]{Updates combo stabs.}
				$A \leftarrow \{i: K_i^\prime \in \mathcal{G}, K_i^\prime = K_aK_i\} \,\cup\, \{i: K_i^\prime \in \mathcal{G}, K_i^\prime = MK_aK_i\},$ \;
				$B \leftarrow \{i: K_i^\prime \in \mathcal{G}, K_i^\prime = MK_i\} \,\cup\, \{i: K_i^\prime \in \mathcal{G}, K_i^\prime = MK_aK_i\}$ \;
				$\mathcal{S}^\prime = \{M^{\abs{c\,\cap\, (B \cup \{a\})}}K_a^{\abs{c\,\cap\, A}} S_{c\setminus \{a\}} : S_c \in \mathcal{S}^\prime\}$ \;
				$\mathcal{S}^\prime \leftarrow \mathcal{S}^\prime \cup \{S_c : S_{c \cup \{a\}} \in \mathcal{S}, S_c \notin \mathcal{S}\}$ \;
				$\mathcal{S} \leftarrow \mathcal{S}^\prime \setminus \textsc{FindTrivStabs}(\mathcal{S}^\prime)$ \;
				$\mathcal{G} \leftarrow \textsc{DetrivialiseGens}(\mathcal{G^\prime} \cup \{M\})$ \tcp*[r]{Detrivs.\ any trivialised gens.} 
			}
			\vspace{5pt}
			\meth{$\textsc{FindMeasurementPatterns}(O, w)$}
			{
				$\mathcal{L}_w \leftarrow \{S_c\bar{L} : S_c \in \mathcal{S}, \bar{L} \in \{\bar{X}, \bar{Z}, \bar{Y}\}, \abs{\qsupp{\bar{L}}} \leq w + 1\}$ \tcp*[r]{Get low-weight logical ops.}
				$\mathcal{M}_w \leftarrow \{M_{\bar{L}_i, \bar{L}_j} : \bar{L}_i, \bar{L}_j \in \mathcal{L}, \acomm{\bar{L}_i^{[O]}}{\bar{L}_j^{[O]}} = 0,  \comm{\bar{L}_i^{[a]}}{\bar{L}_j^{[a]}} = 0 \;\forall\; a \neq O\}$ \tcp*[r]{Find mnt.\ pats.}
				\KwRet $\mathcal{M}_w$
			}
			\vspace{5pt}
		}
		\caption{The $\textsc{StabilizerState}$ object used to implement stabilizer pathfinding.}
		\label{alg:SPF}
	\end{algorithm}
\end{figure}
\clearpage
}

\section{Further results} \label{apx:results}

\subsection{Unheralded loss tolerance for smaller lattice channels} \label{apx:results:UH_thresholds}

Figure \ref{fig:H_vs_UH_channel_thresholds} depicts the comparison between heralded SPF and unheralded SPF for $2 \times 2$, $3 \times 3$ and $4 \times 4$ lattice channels.
In the unheralded case, no threshold crossing is observed and the unheralded teleportation rate is found to decrease with increasing lattice size.
Unlike the heralded loss case, for unheralded loss no clear threshold crossing point is observed in the teleportation rate of different sized lattices.
As with the heralded case, the small size of lattice investigated means that these results are not conclusive.
However, unlike the performance of heralded GPF, in this case the form of $T(p_l)$ remains sigmoidal, and so it is unclear whether such results indicate no threshold exists, or whether it exists but at $p_l \approx 0$.
Given that $p^*_l = 1$ for unheralded loss on the crazy graph, we therefore conjecture that unheralded thresholds do exist (even if they occur at $p_l^* = 0$ on some lattices).
If true, this suggests it may be possible to achieve $p^*_l > 0$ using larger lattices, improved measurement strategies or some lattice structure not considered here.

One possible hypothesis is that all lattice channels exhibit a threshold in the heralded SPF case, but suffer a drop in threshold when loss is unheralded, such that for all the cases considered $p_l^* \rightarrow 0$.
To assess this hypothesis it may be possible to find some lattice channel with a high $p_l^\ast$ in the heralded case with $p_l^* > 0$ when loss is unheralded.
Alternatively, one could attempt to tune between both thresholds (or between threshold and non-threshold behaviour in the case of the null-hypothesis) by simulating intermediately heralded loss where only some fraction of loss is unheralded.

Regardless of whether an unheralded threshold exists or not, it is perhaps unsurprising that SPF under unheralded loss exhibits different behaviour than the heralded case.
In analogy with a quantum error correction protocol consisting of distinct detectability (identifying erroneous qubits) and correctability (calculating some correction operator to apply) substages, SPF under heralded loss has a trivial detectability stage followed by a correctability problem solved over the global state, for which we similarly find a threshold.
On the other hand, when loss is unheralded one cannot separate detectability and correctability into different problems, but rather SPF must solve them simultaneously and with only partial, time-ordered knowledge of the state.
In this case it is therefore not surprising if the phenomena of the heralded case cannot be straightforwardly recovered.
Further study is therefore required to fully understand the differences and similarities of these cases.

\begin{figure}[t]
	\centering
	\vspace{0pt}
	\includegraphics[width=\textwidth]{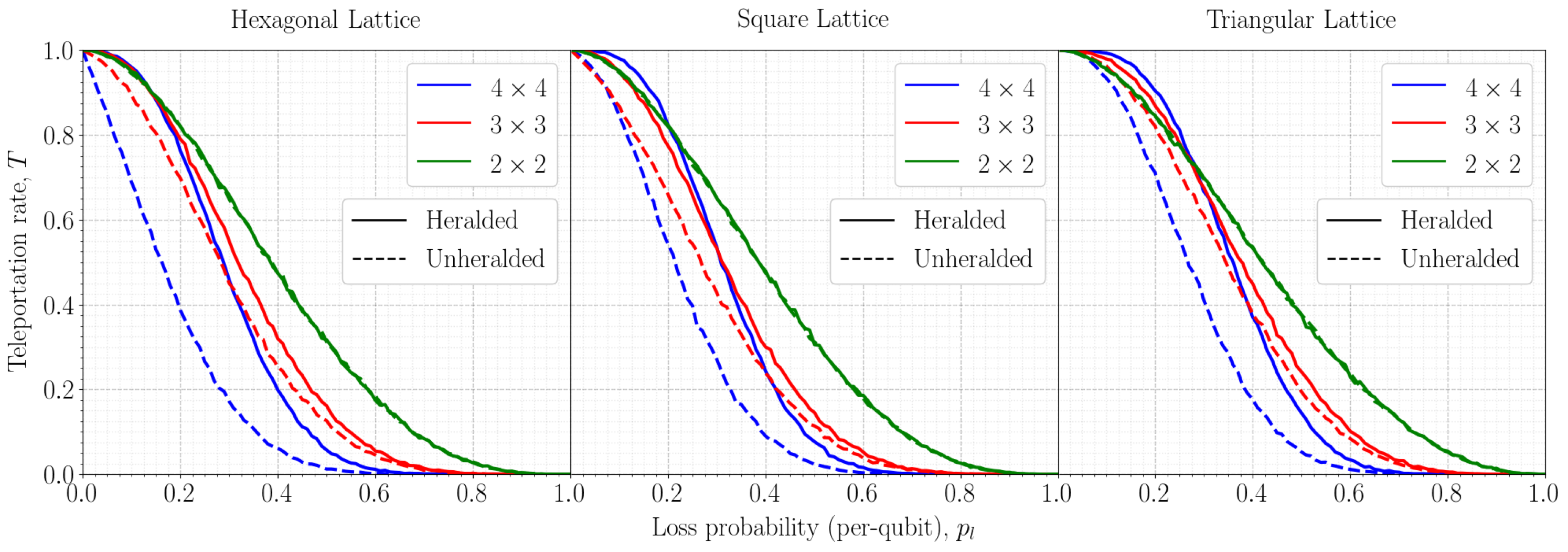}
	\vspace{-15pt}
	\caption{
		Comparison between loss tolerance threshold behaviour for stabilizer pathfinding in the unheralded and heralded regimes for each lattice channel considered (with the ``max-tolerance'' measurement strategy applied for unheralded loss).
	}
	\label{fig:H_vs_UH_channel_thresholds}
	\vspace{0pt}
\end{figure}

\subsection{Measurement strategies} \label{apx:results:mnt_strategies}

\afterpage{
\begin{box_fig}[H]
\begin{framed}
	{\bf UNHERALDED LOSS MEASUREMENT STRATEGY:} \vspace{0pt}
	\begin{enumerate}[label=\textbf{\arabic*}),leftmargin=*]
		\itemsep0em 
	  	\item {\bf Initialise the set of performed measurements as empty $\widetilde{M} = \emptyset$ and the set of available patterns to the set of valid measurement patterns $\widetilde{\mathcal{M}} = \mathcal{M}$.}
	  	\item {\bf Identify some subset of available measurement patterns $\widetilde{\mathcal{M}}^* \subseteq \widetilde{\mathcal{M}}$.}
		\item {\bf Attempt the most common single-qubit measurement $P_i$ across all $M \in \widetilde{\mathcal{M}}^*$.} \\
			In the case of multiple such $P_i$, pick one at random.
		\begin{enumerate}[label=\textbf{\alph*}), leftmargin=*]
			\item {\bf If measurement $P_i$ succeeds (i.e.\ qubit $i$ is not lost), discard all measurement patterns that are no longer available.
			       If no patterns remain, teleportation succeeds.} \\
				Such that $\widetilde{M} \mapsto \widetilde{M} \cup \{P_i\}$ and $\widetilde{\mathcal{M}} \mapsto \{M : P_i \in M \;\forall\; M \in \widetilde{\mathcal{M}}\}$.
			\item {\bf If measurement $P_i$ fails (i.e. qubit $i$ is lost), discard all measurement patterns that contain any measurement on qubit $i$.
			       If no patterns remain, teleportation fails.} \\
				Such that $\widetilde{\mathcal{M}} \mapsto \{M : X_i, Y_i, Z_i \notin M, \;\forall\; M \in \widetilde{\mathcal{M}}\}$.	
		\end{enumerate}
		\item {\bf Repeat from 2).}
	\end{enumerate}
	\vspace{-10pt}
\end{framed}
\vspace{-5pt}
\caption{Algorithm for finding teleportation measurement patterns in case of unheralded loss. (}
\label{box:unheralded_mnt_strategy}
\end{box_fig}
\vspace{0pt}

\begin{figure}[H]
	\centering
	\vspace{0pt}
	\includegraphics[width=\textwidth]{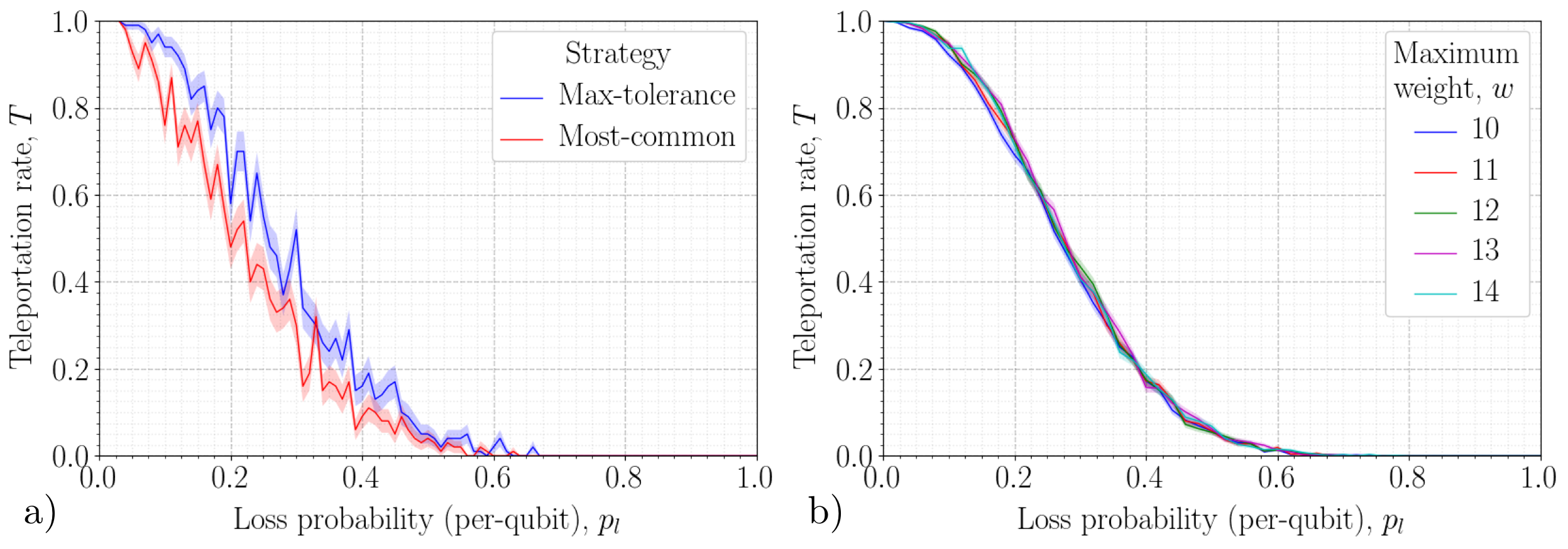}
	\vspace{-15pt}
	\caption{
		a)
		Teleportation rates under unheralded loss for the most-common and max tolerance strategies on the $4 \times 4$ triangular lattice.
		Each data point depicts the average teleportation rate over 100 Monte Carlo instances and error bars are plotted at one standard deviation.
		Measurement patterns were found from pairs of logical operators with at most five weight greater than the minimum.
		Such results show that the prioritisation of the lowest-weight measurement patterns is preferred for unheralded loss over a greater selection of possible measurements.
		We note that the most-common strategy requires approximately $6\times$ longer to compute given a greater number of measurement patterns must be considered.
		b)
		Teleportation rates for stabilizer pathfinding on the $4 \times 4$ triangular lattice with unheralded loss depicted across various measurement pattern maximum weights.
		Recall that $\mathcal{M}_w$ found to a higher maximum weight $w$ will contain increasingly more measurement patterns, although with decreasing loss tolerance.
		Each data point depicts the average value for 1000 Monte Carlo instances and shaded error regions are depicted for a single standard deviation.
		While a small increase in loss tolerance is achieved by an increase above the minimum logical operator weight $w = 10$ to $w = 11$, above this any advantage is marginal at best.
		Such results support the conjecture that the majority of \emph{useful} loss tolerance is provided by a small selection of maximally loss-tolerant measurement patterns produced by $\mathcal{M}_w$ with the minimum $w$.
	}
	\label{apx:fig:max_weight_loss_tol}
	\vspace{0pt}
\end{figure}
\clearpage
}

A general measurement strategy algorithm for teleportation with unheralded loss is given in box \ref{box:unheralded_mnt_strategy}.
Based on the particular choice of $\widetilde{\mathcal{M}}^*$, we present two possible measurement strategies and compare their ability to tolerate unheralded loss.
The \emph{most-common} strategy performs the measurement that occurs most in \emph{all} available measurement patterns, such that $\widetilde{\mathcal{M}}^* = \widetilde{\mathcal{M}}$, whereas the \emph{max tolerance} strategy performs the measurement that occurs most in the \emph{most loss-tolerant} available measurement patterns such that $\widetilde{\mathcal{M}}^* = \{M : \abs{M} = w,\;\forall\; M \in \widetilde{\mathcal{M}}\}$, where $w$ is the minimum measurement pattern weight taken over all $M \in \widetilde{\mathcal{M}}$.

Figure \ref{apx:fig:max_weight_loss_tol}a) compares the performance of the two strategies on a $4 \times 4$ triangular lattice.
Figure \ref{apx:fig:max_weight_loss_tol}b) depicts the performance of the max-tolerance strategy with access to a greater number of measurement patterns as produced by pairs of logical operators with greater maximum weight.

\subsection{Algorithm efficiency} \label{apx:results:time_scaling}

Figure \ref{apx:fig:build_time_scaling} depicts the average computational runtime for building $\ket{\Psi}$ and the finding of associated measurement patterns $\mathcal{M}$ using the algorithms described in Appendix \ref{apx:SPF}.

For building $\ket{\Psi}$, algorithm runtime is primarily a factor of the number of non-trivial stabilizers for the state $\abs{\mathcal{S}_\Psi^{\textrm{NT}}}$.
From this it is easy to see that as $m$ rises, so does the multiplicity of possible generator combinations, with increased $n$ further providing additional qubits to distribute support among.
However, here we observe a drop in build runtime occurs for the highest $m$ when $n \geq 10$.
We conjecture this phenomena is explained by noting that any stabilizer produced from an even number $k$ of generators (where $K_i = X_i \bigotimes_{i \neq j} Z_j$) is trivial for $k>2$ as $K_iK_j = Y_iY_j$ on the completely connected graph of $n$ vertices $K_n$ (given the conventional choice of $\mathcal{G}_\Psi$ for graph states).
Each $k$-clique (i.e.\ $k$-node complete subgraph) within the graph will also have this property.
Hence, as $G_{n,m}$ approaches $K_n$, the number of and size of cliques increases, hence decreasing the number of non-trivial $S_c$ for even $\abs{c} > 2$.

For finding $\mathcal{M}$, we conversely observe an increase in runtime for the most connected graphs.
In this case, the number of near-minimum weight logical operators is the primary factor in the algorithm's runtime.
If the previous conjecture holds, this may also explain the observed increase in runtime here.
As connectivity increases, the number of low-weight $S_c$ associated with cliques also rises and so the number of low-weight logical operators would be expected to increase and hence so to the possible pairings tested for $\mathcal{M}$.
For graphs with near-maximum $m$ it may therefore be sufficient to reduce a search of logical operator pairs to those with absolutely minimal weight.

Finally, we note that the runtime in building $\ket{\Psi}$ depends to some extent on the construction order of the edges.
For example, on the lattices considered in section \ref{sec:loss_tol}, building edges within a vertical layer before building edge between layers was found to be decrease runtime.
While a deep analysis of such optimal construction strategies is beyond the scope of this paper, we conjecture that construction techniques that build highly connected subgraphs first (which are later connected) is preferred to sequentially adding edges to a single growing component.

\begin{figure}[t]
	\centering
	\vspace{-10pt}
	\includegraphics[width=0.67\textwidth]{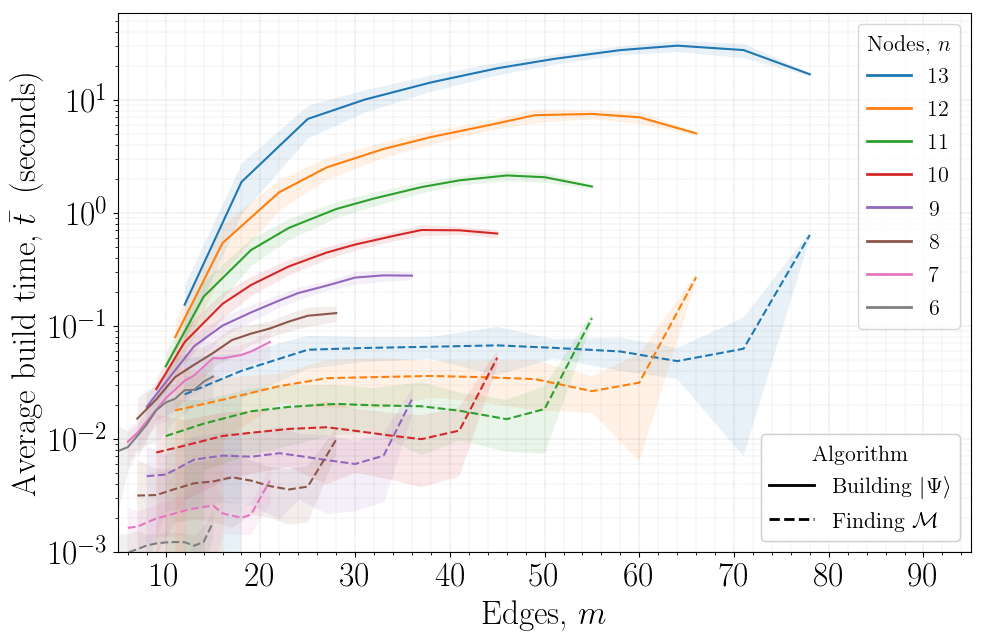}
	\vspace{-5pt}
	\caption{
		Average algorithm runtime for building states, i.e. finding and updating all non-trivial stabilizers (solid lines), and finding measurement patterns from said stabilizers (dashed line).
		Each data point depicts the average run-time for the given algorithm applied to teleportation across 1000 instances of random $n$-qubit graph states with measurement patterns only found from pairs of logical operators with a weight at most three greater than the minimum.
		Specifically, each graph generated is an instance of an Erd\H{o}s-Renyi $G_{n,m}$ random connected graph with $n$ nodes, $m$ edges and $n - 1 \leq m \leq n (n - 1) / 2 - 1$ (also ensuring no edge connecting input qubit $I$ and output $O$).
		Simulations were performed with Python using a standard PC running a 2.8 GHz Intel Core i7 CPU with 16GB of RAM and leveraging a NVIDIA GeForce GT 750M graphics card for GPU processing.
	}
	\label{apx:fig:build_time_scaling}
	\vspace{-5pt}
\end{figure}

\subsection{Configuration loss tolerance} \label{apx:results:config_loss_tol}
\begin{figure}[H]
	\centering
	\vspace{-10pt}
	\includegraphics[width=0.955\textwidth]{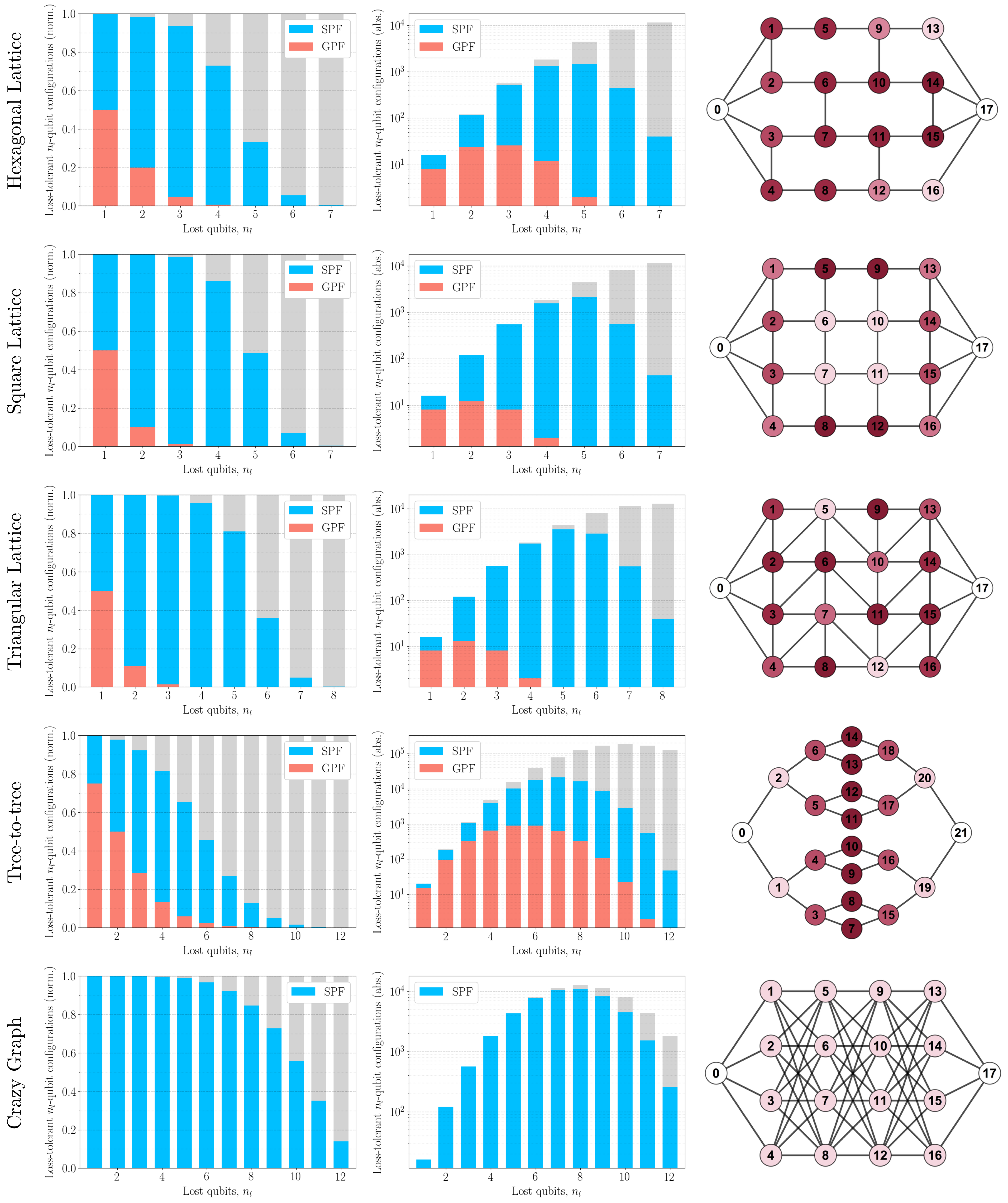}
	\vspace{-5pt}
	\caption{
		The proportion (left) and absolute number (center) of $n_l$-qubit loss configurations tolerable for each considered channel (right) with both GPF (blue) and SPF (red) compared to the total number of $n_l$ qubit configurations (grey).
		Note that the scales of the right-hand plots are logarithmic.
		The total number of $n_l$-qubit configurations for each $n_l$ is shown in grey, given by $\binom{N}{n_l}$, where $N$ is the total number of channel qubits (excluding input and output qubits, which are assumed to be lossless).
		For crazy graph, no GPF loss tolerance exists, and so such data points are omitted.}
	\label{apx:fig:config_loss_tol}
	\vspace{-15pt}
\end{figure}

\end{appendices}

\end{document}